\documentclass[12pt]{article}
\usepackage{booktabs} 
\usepackage{algorithm}
\usepackage{algpseudocode}

\usepackage[utf8]{inputenc} 
\usepackage[T1]{fontenc}    

\usepackage[utf8]{inputenc}
\usepackage[T1]{fontenc}

\usepackage[margin=1in]{geometry}
\usepackage{setspace}
\usepackage{titlesec}
\titlespacing*{\subparagraph}{0pt}{1.25ex plus 1ex minus .2ex}{1em}

\usepackage{amsmath, amssymb, amsfonts, amsthm, amsbsy}
\usepackage{mathtools}
\usepackage{nicefrac}
\usepackage{bbm}           
\usepackage{pifont}        

\usepackage{graphicx}
\usepackage{subcaption}
\usepackage{xcolor}
\usepackage{transparent}
\usepackage{tikz}
\usetikzlibrary{arrows.meta, decorations.pathreplacing, calc}
\usepackage[table]{xcolor}
\usepackage{tikz}
\usepackage{graphicx} 
\usepackage{enumitem}

\usepackage{booktabs}
\usepackage{multirow, bigdelim}
\usepackage{makecell}
\usepackage{array}

\usepackage{natbib}

\usepackage{thmtools} 

\declaretheoremstyle[
  headfont=\bfseries, 
  bodyfont=\itshape, 
  spaceabove=6pt, 
  spacebelow=6pt,
  headpunct={.}\;
]{thmstyle}

\declaretheorem[style=thmstyle,name=Theorem]{theorem}
\declaretheorem[style=thmstyle,name=Proposition]{proposition}

\declaretheorem[style=thmstyle,name=Definition]{definition}

\declaretheorem[style=thmstyle,name=Assumption]{assumption}

\usepackage{url}
\PassOptionsToPackage{hyphens}{url}
\usepackage[
    colorlinks=true,
    linkcolor=blue,      
    citecolor=black,       
    urlcolor=teal        
]{hyperref}

\usepackage[noabbrev, capitalize, nameinlink]{cleveref}
\crefname{assumption}{Assumption}{Assumptions}

\DeclareMathOperator*{\argmin}{arg\,min}

\newcommand{\Var}{\operatorname{Var}}

\usepackage{xspace}
\newcommand{\RR}{Status Quo Reference Rule\xspace}
\newcommand{\GRR}{Reference Rule\xspace}
\newcommand{\GRRs}{Reference Rules\xspace}
\newcommand{\AR}{Anchoring Rule\xspace}

\makeatletter
\newcommand\footnoteref[1]{\protected@xdef\@thefnmark{\ref{#1}}\@footnotemark}
\makeatother

\title{
 Minimizing Volatility: Optimal Adjustment with \\ Evolving Feasibility Constraints
 }
\author{
Simon Jantschgi\thanks{Univ. Zurich, 8050 Zurich, Switzerland; simon.jantschgi@uzh.ch},
\ Heinrich H.~Nax\thanks{Univ. Zurich, 8050 Zurich, Switzerland; heinrich.nax@uzh.ch},
\ Bary S.R.~Pradelski\thanks{CNRS, Maison Fran\c{c}aise d'Oxford and Dept. of Economics, Oxford, OX2 6SE Oxford, UK; bary.pradelski@cnrs.fr}, 
\ Marek Pycia\thanks{Univ. Zurich, 8006 Zurich, Switzerland; marek.pycia@econ.uzh.ch}
	}
\date{January 2026}

\begin{document}

\maketitle
\renewcommand\thefootnote{}
\footnotetext{We are grateful for comments and suggestions by Paul Klemperer and Alex Teytelboym. 
This project has received funding from the European Research Council (ERC) under the European Union's Horizon 2020 research and innovation program grant agreement No 866376 (MP and SJ), from the Swiss National Science Foundation (SNSF) under Eccellenza Grant PCEFP1-181111 (SJ and HHN), by the French National Research Agency through PEPR project FOUNDRY, ANR23-PEIA-0003 (BSRP).}
\renewcommand\thefootnote{\arabic{footnote}}
\vspace{-0.75cm}
\begin{abstract}
Minimizing volatility and adjustment costs is of central importance in many economic environments, yet it is often complicated by evolving feasibility constraints.
We study a decision maker who repeatedly selects an action from a stochastically evolving interval of feasible actions in order to minimize either average adjustment costs or variance. We show that for strictly convex adjustment costs (such as quadratic variation), the optimal decision rule is a reference rule in which the decision maker minimizes the distance to a target action. In general, the optimal target depends both on the previous action and the expectation of future constraints; but 
for the special case where the constraints follow a random walk, the optimal mechanism is to simply target the previous action. If the decision maker minimizes variance, the optimal policy is also a reference rule, but the target is a constant, 
which is not necessarily equal to the long-term average action. 
Compared to mid-point heuristics, these optimal rules may substantially reduce quadratic variation and variance, in natural environments by $50\%$ or more. Applied to stock market auctions, our results provide an explanation for the wide-spread use of reference price rules. We also apply our results to bilateral trade in over-the-counter markets, capacity planning in supply chains, and positioning in political agenda setting.
\end{abstract}

\newpage

\section{Introduction}\label{sec:introduction}

Minimizing volatility and adjustment costs is important in many economic environments.
In asset markets, price volatility creates uncertainty \citep{bansal2004risks} and adjustment costs prevent efficient rebalancing \citep{GrossmanLaroque1990}. Adjustments costs affect capital investment decisions \citep{dixit-pindyck1994investment,CooperEtAl2006} as well as inventory and production management \citep{holt1960planning,lucas1967adjustment,hamermesh1996adjustment}. 
Policy reversals may erode credibility of policy makers \citep{BarroGordon1983, FernandezRodrik1991}. Theoretically, \cite{dekel2025comparative} recently link adjustment costs to a general Le Chatelier principle \citep{Samuelson1947}, where long-run responses to shocks exceed those in the short run.

In practice, not all adjustments are feasible as the choice set is constrained.
Exogenous factors and past decisions shape which actions are feasible.
In repeated double auctions \citep{Wil85,budish2015high}, as used by major stock exchanges, the auctioneer selects a price from a market-clearing interval determined by the overlap of buy and sell orders.
The chosen price and exogenous noise influence future traders' behavior, thereby shifting the clearing interval to the next period. In inventory and production management, firms adjust output subject to supply and demand constraints that might change depending on past production \citep{becker1983simple,AbelEberly1994}. Policy choices may need to remain within a range that is consistent with electoral or institutional viability, and yesterday's policy choice might shape what positions are considered acceptable today \citep{Krugman1991,AlesinaTabellini1990}.

We develop a general theory of optimal decisions in these environments. We study a decision maker who repeatedly selects an action from a time-varying interval of feasible choices that is stochastic and may depend endogenously on the previous period's action.  
The primary objective of the decision maker is to minimize adjustment costs. 
We assume that the adjustment costs are strictly convex, a standard modeling choice in dynamic settings that penalizes large changes more than small ones.
For example, the cost might take the form of quadratic variation, a popular measure of volatility.  
We hence address the question: 
\emph{What course of action minimizes adjustment costs when feasibility constraints evolve?}
This question involves an intertemporal trade-off: 
a myopic policy that tries to repeat the last action reduces immediate adjustment costs but risks large forced adjustments as future feasible intervals may be far from current levels. 
On the other hand, gradually steering actions toward a long-run target limits global dispersion at the cost of frequent, but small, adjustment costs. 

\medskip
\textbf{Illustrative example.}
Consider the following stochastic environment. 
In each period $t$, the random feasible interval 
is generated by the order statistics of two independent draws from a uniform distribution on $[0,1]$ (and hence is independent of past actions). 
Natural benchmark policies include selecting the \emph{Midpoint} of the interval, an \emph{\AR} that minimizes the distance to a fixed reference point, and a ~\emph{\RR} rule that in a given period chooses the action from the feasible interval  that minimizes the distance from the previous period's action (that is, the reference action).\footnote{For financial markets, the {\it Midpoint Rule} corresponds to the often-studied $\frac{1}{2}$-Double Auction of \cite*{chatterjee1983bargaining,Wil85,Rus94}. 
The \emph{\AR} yields a constant action path absent constraints and hence in this case minimizes any measure of adjustment cost.  
The \emph{\RR} rule is employed by stock exchanges in practice \citep*{Jantschgi25,NYSE_Rule_7.35C,NASDAQ_Options_3}.} \cref{fig:policy_comparison} illustrates these policies with an exogenously shifting feasible interval. 
\\

\vspace{-0.25em}
\begin{figure}[ht]
\centering
\scalebox{0.85}{
    \begin{tikzpicture}[
        >=stealth,
        thick,
        x=2.2cm, 
        y=5.4cm,
        axis/.style={->, gray!60, thin},
        tick label/.style={font=\sffamily\small, gray!80!black},
        grid line/.style={gray!20, thin},
        midpoint line/.style={blue!70!black, line width=1.5pt},
        inertia line/.style={red!80!black, line width=1.5pt},
        static line/.style={green!60!black, line width=1.5pt},
        interval line/.style={line width=3.5pt, gray!30, line cap=butt},
        interval cap/.style={line width=1.5pt, gray!60},
        legend box/.style={fill=white, draw=gray!30, rounded corners, inner sep=5pt}
    ]

        
        \def\Lone{0.4} \def\Rone{0.8} 
        \def\PmidOne{0.6}   
        \def\PinerOne{0.6}  
        \def\PstatOne{0.5}  

        \def\Ltwo{0.65} \def\Rtwo{0.9} 
        \def\PmidTwo{0.775} 
        \def\PinerTwo{0.65} 
        \def\PstatTwo{0.65} 

        \def\Lthree{0.3} \def\Rthree{0.6} 
        \def\PmidThree{0.45} 
        \def\PinerThree{0.6} 
        \def\PstatThree{0.5} 

        \def\Lfour{0.35} \def\Rfour{0.75} 
        \def\PmidFour{0.55} 
        \def\PinerFour{0.6}  
        \def\PstatFour{0.5}  

        \def\Lfive{0.1} \def\Rfive{0.35} 
        \def\PmidFive{0.225} 
        \def\PinerFive{0.35} 
        \def\PstatFive{0.35} 

        \foreach \y in {0.25, 0.5, 0.75, 1.0} {
            \draw[grid line] (0.5, \y) -- (5.5, \y);
            \draw (0.5, \y) -- (0.4, \y) node[left, tick label] {\y};
        }
        
        \draw[axis] (0.5, 0) -- (5.5, 0) node[right, black, font=\small] {\textbf{Time} $t$};
        \draw[axis] (0.5, 0) -- (0.5, 1.1) node[above, black, font=\small] {\textbf{Action Space}};

        \foreach \t in {1,2,3,4,5} {
            \draw[gray!40] (\t, 0) -- (\t, -0.02);
            \node[below, tick label] at (\t, -0.02) {\t};
        }

        \foreach \t/\L/\R in {1/\Lone/\Rone, 2/\Ltwo/\Rtwo, 3/\Lthree/\Rthree, 4/\Lfour/\Rfour, 5/\Lfive/\Rfive} {
            \draw[interval line] (\t, \L) -- (\t, \R);
            \draw[interval cap] (\t-0.12, \L) -- (\t+0.12, \L);
            \draw[interval cap] (\t-0.12, \R) -- (\t+0.12, \R);
        }

        \draw[midpoint line] plot coordinates {(1, \PmidOne) (2, \PmidTwo) (3, \PmidThree) (4, \PmidFour) (5, \PmidFive)};
        \draw[inertia line] plot coordinates {(1, \PinerOne) (2, \PinerTwo) (3, \PinerThree) (4, \PinerFour) (5, \PinerFive)};
        \draw[static line] plot coordinates {(1, \PstatOne) (2, \PstatTwo) (3, \PstatThree) (4, \PstatFour) (5, \PstatFive)};

        \node[legend box, anchor=north east] at (5.5, 1.25) {
            \begin{tikzpicture}[x=1cm, y=0.5cm] 
                \draw[interval line] (0.25, 0.6) -- (0.25, 0.0);
                \draw[interval cap] (0.1, 0.6) -- (0.4, 0.6);
                \draw[interval cap] (0.1, 0.0) -- (0.4, 0.0);
                \node[anchor=west, font=\sffamily\small] at (0.6, 0.3) {Feasible Set};
                
                \draw[blue!70!black, line width=1.5pt] (0, -0.8) -- (0.5, -0.8);
                \node[anchor=west, font=\sffamily\small] at (0.6, -0.8) {Midpoint};
                
                \draw[red!80!black, line width=1.5pt] (0, -1.6) -- (0.5, -1.6);
                \node[anchor=west, font=\sffamily\small] at (0.6, -1.6) {\RR};
                
                \draw[green!60!black, line width=1.5pt] (0, -2.4) -- (0.5, -2.4);
                \node[anchor=west, font=\sffamily\small] at (0.6, -2.4) {\AR at $0.5$};
            \end{tikzpicture}
        };

    \end{tikzpicture}
}
\caption{\textbf{Comparison of Benchmark Policies.} The graph traces the actions selected by the \emph{Midpoint Rule} (blue), the \emph{\RR} (red), and the \emph{\AR} rule (green). The gray bars indicate the feasible interval in a given period.}
\label{fig:policy_comparison}
\vspace{-0.25em}
\end{figure}
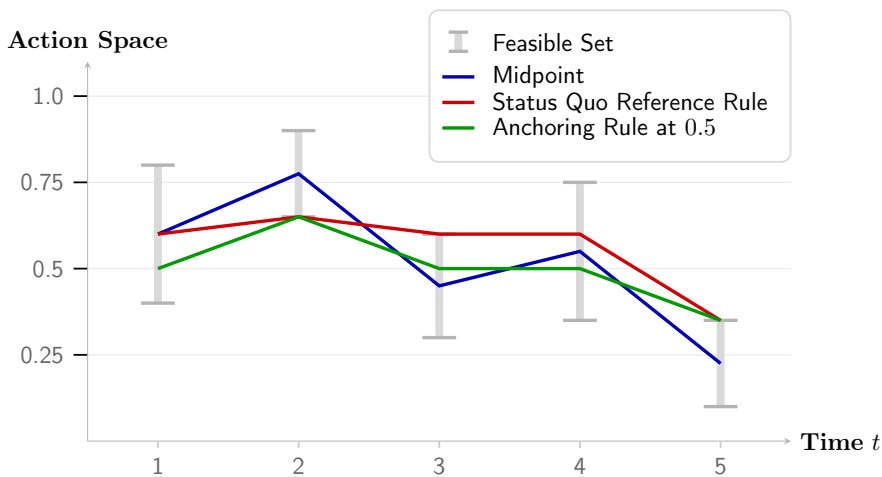

Each of these policies is distinct and may lead to more or less fluctuation in actions. 
We evaluate the three benchmark policies against \emph{quadratic variation:} 
The adjustment cost incurred in any period is the squared change compared to the previous period. 
The total cost is the sum of the period-by-period costs. 

It turns out that both the ~\AR and the ~\RR  yield substantive gains over the Midpoint Rule, reducing quadratic variation by  50\% or more.
For the analytical derivation of the resulting quadratic variation for each rule, see \cref{app:uniform-benchmark}. 
Our analysis shows that optimal rules have a similar, simple structure, even when the feasible intervals evolve endogenously.

\subsection{Summary of Results}

We consider a broad class of stochastic environments where feasibility constraints evolve randomly and endogenously. We model this evolution as the set-valued analogue of an autoregressive process of order 1 (AR(1)): just as a standard autoregressive process describes how a single point evolves based on its history, our framework describes how a \emph{feasible interval ($I_t$)} evolves over time ($t$).
This dynamic structure captures three distinct economic forces: \emph{endogeneity}, where the future interval is stochastically anchored to the agent's previous choice (captured by $C_t(P_{t-1})$); \emph{exogenous drift}, representing fundamental, stochastic shocks that shift the interval regardless of past actions (via an additive component); and \emph{stochastic width}, capturing changing levels of discretion or flexibility ($J_t$). See \cref{fig:model_dynamics} for an illustration. By varying the degree of endogeneity, our framework includes purely exogenous constraints (independent of past actions; see the example above) to random walk dynamics (where constraints move one-to-one with actions). 
Within this environment, we analyze a broad class of strictly convex adjustment costs. In particular, we allow cost functions that may be asymmetric and penalize deviations in one direction more than the other.

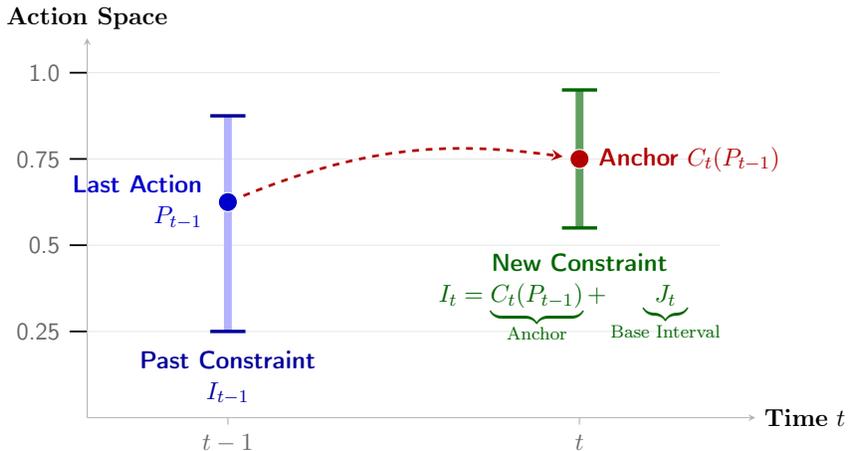
\begin{figure}[ht]
\setlength{\belowcaptionskip}{-12pt}
\centering
\scalebox{0.85}{
    \begin{tikzpicture}[
        >=stealth,
        thick,
        x=2.75cm, 
        y=5.4cm,  
        axis/.style={->, gray!60, thin},
        tick label/.style={font=\sffamily\small, gray!80!black},
        label text/.style={font=\sffamily\small},
        grid line/.style={gray!20, thin},
        price point/.style={circle, fill=blue!80!black, inner sep=3pt, draw=white, line width=0.5pt},
        anchor point/.style={circle, fill=red!70!black, inner sep=3pt, draw=white, line width=0.5pt},
        interval bar old/.style={line width=3.5pt, blue!30, line cap=butt},
        interval cap old/.style={line width=1.5pt, blue!60!black, line cap=butt},
        interval bar new/.style={line width=3.5pt, green!25!gray, line cap=butt},
        interval cap new/.style={line width=1.5pt, green!40!black, line cap=butt},
        transition arrow/.style={->, dashed, red!70!black, very thick, shorten >=3pt, shorten <=3pt}
    ]

        \def\tx{1.0}        
        \def\ty{3.0}        


        \def\yP{0.625}       
        \def\yLold{0.25}     
        \def\yRold{0.875}    
        \def\yAnchor{0.75}   
        \def\yLnew{0.55}     
        \def\yRnew{0.95}     

        \foreach \y in {0.25, 0.5, 0.75, 1.0} { 
            \draw[grid line] (0.2, \y) -- (3.8, \y);
            \draw (0.2, \y) -- (0.1, \y) node[left, tick label] {\y};
        }

        \draw[axis] (0.2, 0) -- (4.0, 0) node[right, black, font=\small] {\textbf{Time} $t$};
        \draw[axis] (0.2, 0) -- (0.2, 1.1) node[above, black, font=\small] {\textbf{Action Space}};

        \draw[gray!40] (\tx, 0) -- (\tx, -0.02) node[below, tick label] {$t-1$};
        \draw[gray!40] (\ty, 0) -- (\ty, -0.02) node[below, tick label] {$t$};


        \draw[interval bar old] (\tx, \yLold) -- (\tx, \yRold);
        \draw[interval cap old] (\tx-0.1, \yLold) -- (\tx+0.1, \yLold);
        \draw[interval cap old] (\tx-0.1, \yRold) -- (\tx+0.1, \yRold);

        \node[below=0.15cm, blue!60!black, label text, align=center] at (\tx, \yLold) 
            {\textbf{Past Constraint}\\$I_{t-1}$};

        \node[price point] (P_prev) at (\tx, \yP) {};
        \node[left=0.25cm, blue!80!black, label text, align=right] at (P_prev) 
            {\textbf{Last Action}\\$P_{t-1}$};


        \draw[interval bar new] (\ty, \yLnew) -- (\ty, \yRnew);
        \draw[interval cap new] (\ty-0.1, \yLnew) -- (\ty+0.1, \yLnew);
        \draw[interval cap new] (\ty-0.1, \yRnew) -- (\ty+0.1, \yRnew);

        \node[below=0.25cm, green!40!black, font=\small, align=center] at (\ty, \yLnew) 
            {\textbf{\sffamily New Constraint}\\
             $I_t = \underbrace{C_t(P_{t-1})}_{\text{\scriptsize Anchor}} + \underbrace{J_t}_{\text{\scriptsize Base Interval}}$};

        \node[anchor point] (C_curr) at (\ty, \yAnchor) {};
        \node[right=0.15cm, red!70!black, label text, align=left] at (C_curr) 
            {\textbf{Anchor} $C_t(P_{t-1})$};


        \draw[transition arrow] (P_prev) to[bend left=15] 
            node[midway, above=3pt, align=center, red!70!black, font=\bfseries\sffamily\small] 
            {} 
            (C_curr);
    \end{tikzpicture}
}
\caption{\textbf{Evolution of the Feasible Interval.} The figure illustrates the formation of the constraint at time $t$. First, the agent selects an action $P_{t-1}$ (blue dot) from the vertical interval $I_{t-1}$. This choice determines the location of the new stochastic anchor $C_t(P_{t-1})$ (red dot) via a stochastic  transformation. Finally, the new feasible interval $I_t$ (green bar) is realized by adding a stochastic base interval $J_t$ to the anchor.}
\label{fig:model_dynamics}
\end{figure}

Our first result (\Cref{thm:optimal_policy_adjustment}) states that the optimal policy takes the form of a simple \emph{\GRR}, in which the decision maker minimizes the distance to a target point that is continuous in the previous action.
This is done by projecting the target onto the feasible interval: The decision maker selects the target if it is feasible, the lower endpoint if the target falls below the interval, and the upper endpoint if it lies above.
 
\Cref{thm:optimal_policy_adjustment} also shows that it is optimal to disregard a significant portion of the information available to the decision maker. 
A priori the target point could depend on the current period draw of the constraint interval (or, in applications, other current period realizations of uncertainty), 
as it does, for instance, in Midpoint Rules. 
The theorem, however, establishes that it is sufficient to condition the target point only on the information from previous periods. 
This and our next results thus provide an explanation for the popularity of Reference Rules---rather than Midpoint Rules---in many stock market auctions. 

In \Cref{theorem:non-expansive}, we show that for a broad class of adjustment costs the optimal policy is strictly increasing and slowly drags the decision maker's target point back to a fixed center. The optimality of such {\GRRs} is driven by 
the fundamental trade-off imposed by strictly convex adjustment costs: 
because the cost of a change increases disproportionately with its size, two small adjustments are always cheaper than one large adjustment of the same total magnitude. 
Consequently, a myopic policy that minimizes immediate adjustment costs by repeating the previous action may be suboptimal. 
While it avoids costs today, it risks leaving the decision maker ``cornered'' by future constraints, necessitating a large, expensive correction later.
Instead of reacting only when forced, the decision maker may want to anticipate future constraints and actively move towards them gradually. 
This effectively distributes the necessary adjustment across multiple periods, smoothing the path and minimizing total adjustment costs.\footnote{Suppose the decision maker is currently forced to select a low action, $P_t \approx 0$, but expects future constraints to shift to the right. If the next interval allows it (e.g., $I_{t+1} = [0,1]$), a myopic agent would stay at $0$. However, if the interval later shifts to $I_{t+2} = [0.5, 1]$, this agent is forced to jump $0 \to 0.5$ all at once (cost $\propto 0.5^2 = 0.25$). The optimal reference rule anticipates this and starts moving early, selecting an intermediate step like $0.25$ in period $t+1$. This path ($0 \to 0.25 \to 0.5$) incurs a total cost proportional to $0.25^2 + 0.25^2 = 0.125$, cutting the total adjustment burden in half.} 

\Cref{theorem:rwinertia} states that in the random walk limit --- where the location of the feasible interval moves one-to-one with the previously chosen action --- if the adjustment cost depends only on the size of the change and not the absolute action level, the \RR is the unique optimal policy. 
In this case, the optimal target is simply the last action.
The reason is that the random walk environment is translation invariant. 
Unlike a mean-reverting case, where the decision maker navigates toward a fixed attracting center to minimize the impact of future constraints, here the constraints move with the agent. The complex infinite-horizon problem effectively decouples into a sequence of static choices: at every step, simply minimize the immediate adjustment cost. 

We next turn to the objective of minimizing the variance of the action path, that is, minimizing the squared distance of the actions to their mean.\footnote{This objective violates some of the assumptions made for \Cref{thm:optimal_policy_adjustment}.}
\Cref{theorem:varanchor} states that for this objective, the optimal policy is an~\AR that minimizes the distance to a constant anchor. 
Intuitively, the variance objective penalizes \emph{action dispersion} rather than the size of the adjustment. 
Since the decision maker faces no penalty for the period-by-period magnitude of the change, the intertemporal smoothing motive vanishes.
The agent is willing to make large jumps to return to the attracting position in the state space immediately, rather than drifting there gradually.
Crucially, because constraints are endogenous, the optimal anchor is not necessarily the long-run average action. 
To keep the process centered, the policy may aim at a different value to counteract the natural drift of the constraints.
For exogenous environments, we characterize this anchor via a first-order condition (\Cref{theorem:optanchor}). The optimal anchor must equalize the expected distortion from the lower and upper constraints.
In our uniform example, the~\AR at $\frac{1}{2}$ is globally optimal.

In \Cref{sec:applications}, we consider several applications that are nested within our general framework.  
First, for \emph{repeated double auctions}, \Cref{theorem:rwinertia} rationalizes the prevalence of reference prices: if the market fundamental approximates a random walk, the \RR found on stock exchanges is indeed optimal for minimizing transaction volatility and beats the standard midpoint rules studied in theory.
Second, in \emph{OTC markets}, we examine the incentive properties of a \GRR within a standard model of bilateral trade, where the feasible interval is not exogenous but set by strategic counterparties. The reference price creates a sharp non-convexity between a monopsony and a price-taking regime, which precludes the existence of standard Bayes-Nash equilibria in increasing strategies, but allow for a pooling equilibrium similar to a posted price mechanism at the Reference price. This may explain price-stickiness in OTC markets.
Third, for \emph{capacity planning} in supply chains, we illustrate how optimal production smoothing acts as a buffer against simultaneous supply and demand shocks.
Finally, in \emph{political agenda setting}, we demonstrate that a policymaker's optimal strategy bifurcates depending on their objective: minimizing ``flip-flopping'' (adjustment costs) leads to path-dependent policy adjustments (\Cref{thm:optimal_policy_adjustment}), while minimizing ideological inconsistency (variance) dictates a static agenda (\Cref{theorem:varanchor}). 

Our analysis is based on the framework of stochastic dynamic programming and Markov decision processes in discrete time 
\citep{Puterman1994,bertsekas1996stochastic,HernandezLermaLasserre1999}. 
In modeling the evolution of feasibility, we build on the theory of 
autoregressive dynamics \citep{hamilton2020time}. 
Most of the relevant economic literature has been developed within the context of specific application, which we discuss in the above-mentioned \Cref{sec:applications}. 

\section{The Model}
\label{sec:model}

We study the problem of a decision maker who seeks to stabilize a stochastic process that is subject to \emph{time-varying} and \emph{endogenous} feasibility constraints. At the beginning of each period $t=1,2,\dots$, the decision maker observes a \emph{feasible interval} $I_t = [L_t, R_t] \subset \mathbb{R}$.\footnote{
We define the state space as the entire real line for analytical convenience, as it avoids the technical complexities of boundary conditions.
For many applications, such as political positioning, this seems natural.
In the specific context of prices, the assumption of an unbounded (and negative) support is appropriate for two reasons.
First, in many applications, $P_t$ represents the \emph{logarithm} of an asset price, a basis spread, or a net inventory position, all of which naturally span $\mathbb{R}$. Even for nominal prices, negative values have precedents in commodity markets (e.g., electricity or oil futures).
Second, under the stability conditions we impose below, if the drift terms ($B_t$) are sufficiently positive, or the persistance terms ($A_t$) are sufficiently small, the probability of the process visiting negative values can be made arbitrarily small.}

The decision maker must choose an action $P_t$ within this interval, $P_t \;\in\; I_t$.
A core friction of the model is that the constraints are endogenous: the interval available today may depend on the action chosen yesterday. 
This creates a dynamic feedback loop where current stabilization efforts constrain future flexibility.

\subsection{The Evolution of the Feasible Interval}

We model the evolution of the feasible interval as a \emph{stochastic linear transition process}.

\smallskip
\begin{assumption}[Transition Dynamics]\label{assum:transitiondynamics}
The evolution of the feasible interval $I_t$ is generated by  a random \emph{base interval} $J_t = [U_t, V_t]$ drawn i.i.d.\ from a fixed distribution and shifted  by a location parameter dependent on the previous action. Specifically:
\begin{equation}
\label{eq:transition-affine}
I_t \;=\; \left[L_{t},R_{t}\right] \;=\;C_t(P_{t-1}) \;+\; J_t\;=\; \left[C_t(P_{t-1}) \;+\; U_t\;,\;C_t(P_{t-1}) \;+\; V_t\right] \tag{Feasible Interval}
\end{equation}
The stochastic anchor 
\begin{equation}
C_t(P_{t-1}) = A_t P_{t-1} + B_t \tag{Stochastic Anchor}
\label{eq:AR1}
\end{equation}
is an affine function of the previous action, where the coefficients $(A_t, B_t)$ are i.i.d.\ random variables.
\end{assumption}
Equations \eqref{eq:transition-affine}  and \eqref{eq:AR1} are a set-valued analog of a standard AR(1) process; the anchor $C_t$ follows an AR(1) process and the interval is drawn randomly around the anchor. 
To ensure that the control problem is well-posed on an unbounded domain, we impose the following stability conditions on the transition parameters. 

\begin{assumption}[$L^{p}$-Stability]\label{assum:stability}
There exists some $p \geq 2$, such that the scaling factor satisfies $\mathbb{E}[\vert A_t\vert ^p] < 1$ and $\mathbb{E}[|B_t|^p] \;+\; \mathbb{E}[|U_t|^p] \;+\; \mathbb{E}[|V_t|^p] \;<\; \infty.$
\end{assumption}

This assumption guarantees that the system is dissipative and admits a stationary distribution with finite moments up to order $p$. 
Crucially, it also ensures that the expected total cost remains finite even for adjustment costs with polynomial growth of order $p$ (introduced in \cref{assum:cost}). 

\medskip
\noindent The model captures three distinct forces:
\begin{itemize}[leftmargin=*, labelindent=0pt]
\item \emph{Persistence $(A_t)$:} This coefficient measures the \emph{endogeneity} of the future constraint. It determines the degree of path dependence in the system. Two important special cases are \emph{exogenous feasibility} where the feasible interval only moves exogenously ($A_t \equiv 0$) and \emph{random walk} where the feasibly interval is fully path-depedent ($A_t \approx 1$). Because the case $A_{t}=1$ violates \cref{assum:stability}, we treat it separately in \Cref{theorem:rwinertia}.
    \item \emph{Drift $(B_t)$:} This captures exogenous shifts in the \emph{location} of the feasible interval. In our double auction application, for example, this captures new information or fundamental value shocks.
    \item \emph{Feasibility Constraint $(J_t)$:} The random width of the base interval, $|V_t - U_t|$, represents the room for maneuver. A wide interval implies high discretion, while a narrow interval forces the decision maker's hand.
\end{itemize}

\subsection{Admissible Policies}

In each period, the decision maker must select an action $P_t$ from the feasible interval $I_{t}$. 
We adopt the most permissive definition of a strategy, allowing for history-dependence and randomization.

\begin{definition}[Admissible Policy]
Let $\mathcal{F}_t$ denote the information set available at time $t$, comprising of all past actions and constraints. An admissible policy $\pi = \{\nu_1, \nu_2, \dots\}$ is a sequence of decision rules in which $\nu_t$ maps the information $\mathcal{F}_t$ to a probability distribution over $\mathbb{R}$.
A policy is \emph{admissible} if the selected action $P_t$ respects the constraint almost surely: $P_t \sim \nu_t(\cdot \mid \mathcal{F}_t) \implies P_t \in I_t$.
\end{definition}

Many practical policies --- including the canonical examples we  introduce below --- depend only on the current state of the system. This motivates the following definition.

\begin{definition}[Deterministic Stationary Markov Policy]
A policy is a \emph{stationary Markov policy} if it can be represented by a single time-invariant measurable function $\mu: \mathbb{R} \times \mathcal{I} \to \mathbb{R}$ such that:
$P_t \;=\; \mu(P_{t-1}, I_t) \quad \text{almost surely for all } t.$
\end{definition}

\noindent Stationary Markov policies coupled with the independent interval process (\cref{assum:transitiondynamics}) define a time-homogeneous Markov chain for the action $P_t$. 
We say that the policy $\mu$ \emph{induces a stationary distribution} if this Markov chain admits an invariant probability measure $\psi$.\footnote{Formally, if the action at time $t-1$ is distributed according to some distribution, then the action at time $t$ follows the same distribution.}
In \cref{ssec:stability}, we show that \cref{assum:transitiondynamics} ensures that the action system is globally dissipative: under any admissible policy, the moments of the action path remain uniformly bounded. 
Moreover, for stationary Markov policies, this boundedness guarantees the existence of an invariant probability measure.

\smallskip
We next introduce two structural families of policies that serve as intuitive benchmarks:
\begin{itemize}[leftmargin=*, labelindent=0pt]
    \item \emph{Convex-Combination Rules:} The decision maker selects the action within the interval based on fixed relative weights. For $k \in [0,1]$, $P_t \;=\; (1-k)L_t + k R_t$.
    This includes adhering to the lower bound ($k=0$), the upper bound ($k=1$), or the \emph{midpoint rule}  ($k=0.5$). These rules are purely reactive to the constraint.
    \smallskip
    
 \item \emph{Reference Rules:} The decision maker targets an ideal action $r(P_{t-1})$ and moves as close to it as the interval allows. Let $\Pi_{[L,R]}(x) = \min(\max(x, L), R)$ denote the Euclidean projection of $x$ onto the interval $[L,R]$. Then the rule is given by: $P_t \;=\; \Pi_{I_t}\big( r(P_{t-1}) \big)$.
    Two canonical policies within this class are:
    \smallskip
    \begin{itemize}
        \item \emph{\AR:} $r(p) \equiv \mu$. The decision maker consistently targets a fixed fundamental value $\mu$, regardless of past actions, correcting deviations whenever the feasible interval permits.
        \item \emph{\RR:} $r(p) = p$. The decision maker myopically targets the previous action $P_{t-1}$. Under this rule, they maintain the status quo ($P_t = P_{t-1}$) whenever it is feasible, and make the minimal necessary adjustment to satisfy the constraint only when the previous action falls outside the new interval $I_t$.
        \end{itemize}
\end{itemize}

\subsection{Objective: Minimizing Adjustment Costs}

The decision maker aims to minimize the aggregate of period-by-period adjustment costs $c: \mathbb{R} \times \mathbb{R} \to [0, \infty)$ that depends on the previous action $P_{t-1}$ and the current action $P_t$. The objective is to find a policy $\pi$ that minimizes the long-run average cost:
\begin{equation}
\label{eq:general_objective}
\mathcal{C}(\pi) \;:=\; \limsup_{T\to\infty} \frac{1}{T} \sum_{t=1}^T \mathbb E^\pi \big[ c(P_{t-1}, P_t) \big].\tag{Long-Run Adjustment Costs}
\end{equation}

To capture the economic meaning of adjustment costs, we impose the following assumptions on the cost function.

\begin{assumption}[Adjustment Cost Structure]
We assume that the cost function $c(x,Y)$ satisfies:
\begin{enumerate}[leftmargin=*, labelindent=0pt]
    \item \emph{Regularity:} $c(x,y)$ is jointly convex and twice continuously differentiable and strictly convex in $y$.

    \item \emph{No Cost for Inaction:} $c(x,x) = 0$ and $\nabla c(x,x) = 0$.
    
    \item \emph{Monotonicity:} For any fixed past action $x$, the cost is strictly increasing as the current action $y$ moves away from $x$. That is, $\frac{\partial c}{\partial y} > 0$ for $y > x$ and $\frac{\partial c}{\partial y} < 0$ for $y < x$.
    
    \item \emph{Polynomial Growth:} The cost function $c(x,y)$ grows at most polynomially of the same order $p$ defined in  \cref{assum:stability}. Specifically, there exists a constant $K > 0$ such that for all $x, y \in \mathbb{R}$: $c(x, y) \;\le\; K \left( 1 + |x|^p + |y|^p \right)$.\footnote{This ensures that the expected single-period cost is finite.}
\end{enumerate}
\label{assum:cost}
\end{assumption}

\noindent For some applications, we will consider a special class of adjustment costs, that only depend on the size of the adjustment, not the location.

\begin{definition}[Translation Invariant Costs]\label{def:translationinvariant}
We say that adjustment costs are \emph{translation invariant}, if there exists a strictly convex, continuously differentiable function $\phi: \mathbb{R} \to [0, \infty)$ with $\phi(0)=0$ and $\phi'(0)=0$ such that: $c(x,y) \;=\; \phi(y-x)$.
\end{definition}

\smallskip
\cref{assum:cost} encompasses many canonical measures of adjustment costs and volatility. For example,  quadratic cost functions ($c(x,y) \;=\; (y-x)^2$) and the resulting aggregate quadratic variation, which are the standard proxy for volatility and execution risk. Next, cost functions that are strongly convex for small adjustments and less steep for extreme values, such as, for example, the Pseudo-Huber loss function used in robust regression. Finally, cost functions may be asymmetric, for example penalizing downward adjustments less than upward adjustments.

An important special case, which does not satisfy \cref{assum:cost} but appears as a limit case,\footnote{This cost function violates the strict convexity assumption in the second argument.} is the squared error from a fixed reference point, that is, for a constant $M \in \mathbb{R}: c(x,y)=(y-M)^{2}$.
The resulting aggregate objective corresponds to the overall dispersion of the actions and their variance (when $M$ is set to mean). We will study this cost function separately in \Cref{ssec:variance}.
    
\section{The Optimality of Reference Rules}
\label{sec:qv-opt}

We now determine the \emph{globally optimal} mechanism. Adjustment cost minimization fundamentally requires history dependence. If the previous action was $P_{t-1}$, moving to any $P_t \neq P_{t-1}$ incurs a cost. 
Thus, the optimal policy must balance the immediate benefit of staying at (or close to) $P_{t-1}$ against the future risk of large adjustments due to the evolving constraints.
We prove that the optimal mechanism always takes a simple, deterministic structural form: the \emph{\GRR Policy}.

\begin{theorem}[Optimality of Reference Rules]
\label{thm:optimal_policy_adjustment}
Consider the objective of minimizing adjustment costs satisfying \cref{assum:cost}.
There exists a unique optimal stationary policy. 
This policy takes the form of a \GRR: there exists a continuous target function $r: \mathbb{R} \to \mathbb{R}$ such that the optimal action is the projection of the target onto the feasible interval: $P_t^* \;=\; \Pi_{I_t} \big( r(P_{t-1}) \big).$
\end{theorem}

\noindent\textbf{Proof Outline and Intuition.}
The formal proof (provided in \Cref{ssec:ProofTheorem1}) proceeds by casting the minimization of adjustment costs as an infinite-horizon \emph{Markov Decision Process (MDP)}. 
We utilize the machinery of Dynamic Programming to characterize the optimal policy. 

\noindent\textbf{Part 1: The Mathematical Framework.}
The problem is to minimize the long-run average cost $\mathcal{C}(\pi) \;:=\; \limsup_{T\to\infty} \frac{1}{T} \sum_{t=1}^T \mathbb E^\pi \big[ c(P_{t-1}, P_t) \big]$. This is a classic Average Cost MDP with a continuous state space $\mathbb{R}$ and stochastic constraints.
\begin{itemize}[leftmargin=*, labelindent=0pt]
    \item \textbf{The Discounted Problem.} To analyze the system, we first consider the $\beta$-discounted problem, where the objective is to minimize the total present value of future costs:
    \[
    \mathbb{E} \left[ \sum_{t=1}^\infty \beta^{t-1} c(P_{t-1}, P_t) \right], \quad \text{for } \beta \in (0,1). \tag{Discounted Problem}
    \]
    Using a discount factor ensures the total cost is finite. We define the value function $V_\beta(x)$ as the minimum achievable expected cost starting from an initial action $P_0 = x$. This function satisfies:
    \[
    V_\beta(x) \;=\; \mathbb{E}\left[ \min_{y \in I(x)} \big\{ c(x,y) + \beta V_\beta(y) \big\} \right]. \tag{Bellman Equation}
    \]
    We prove that the Bellman operator is a contraction on a weighted Banach space, ensuring that a unique, convex solution $V_\beta$ exists.
    \smallskip
    \newpage

    \item \textbf{The Average Cost Limit (ACOE).} By taking the vanishing discount limit ($\beta \uparrow 1$), we derive the \emph{Average Cost Optimality Equation}, which characterizes the solution to the original problem. The optimal policy minimizes the sum of the immediate adjustment cost and the differential value function $h(y)$, which captures the relative future cost of being in state $y$:
    \[
    \lambda^* + h(x) \;=\; \mathbb{E}\left[ \min_{y \in I(x)} \big\{ c(x,y) + h(y) \big\} \right]. \tag{Average Cost Optimality Equation}
    \]
    In this equation, the constant $\lambda^*$ represents the \emph{minimal long-run average cost}. 
\end{itemize}

\noindent\textbf{Part 2: The Separation Principle.}
The solution to the ACOE reveals a striking structural property: the optimal policy decouples into a \emph{dynamic planning phase} and a \emph{static execution phase}. 

\begin{enumerate}[leftmargin=*, labelindent=0pt]
    \item \textbf{The Optimal Target.}
    First, the decision maker solves the minimization problem as if the current constraint $I_t$ did not exist. They face a strictly convex objective function that combines two opposing forces:
    \begin{itemize}
        \item \emph{Immediate Inertia:} The adjustment cost $c(P_{t-1},P_{t})$ pulls the action toward the previous value to minimize volatility.
        \item \emph{Future Safety:} The value function $h(p)$ encodes the expected cost of future constraints. It pulls the action toward a safer action to avoid large deviations later.
    \end{itemize}
    The unconstrained minimizer of this sum is the \emph{Reference Target} $r(P_{t-1})$. It represents the ideal compromise between immediate inertia and future safety.
    
    \item \textbf{The Feasibility Constraint.}
    Next, the interval $I_t$ realizes. 
    Because the total objective is strictly convex, the solution is straightforward: the decision maker should get as close to the reference point as the interval allows. 
    Mathematically, this is the projection of the target $r(P_{t-1})$ onto the interval $I_t$.
    If the target is feasible, the agent chooses it exactly. If the target is outside the interval, the agent chooses the nearest boundary. \hfill $\diamondsuit$
\end{enumerate}

Next, we show that if adjustment costs are submodular and the marginal cost is at least as sensitive to the current action as to the past action, then the optimal target function is increasing and non-expansive, that is, the target function moves less than one-to-one with the last action.

\newpage

\begin{theorem}[Non-Expansiveness of Optimal \GRRs]
\label{theorem:non-expansive}
Consider the objective of minimizing adjustment costs satisfying \cref{assum:cost}.
If the cost function additionally satisfies $c_{xy}(x,y) < 0$ and $|c_{xy}| \le c_{yy}$ for all $x,y$, the optimal target function is increasing and non-expansive: for any distinct actions $p_1 < p_2$:
$0 \;\le\; r(p_2) - r(p_1) \;<\; p_2 - p_1$. Moreover, there exists a unique fixed point $p^\ast$ with $r(p^\ast) = p^\ast$.
\end{theorem}

The proof is relegated to \Cref{ssec:proof-non-expansive}.
The submodularity and non-expansiveness conditions are automatically satisfied by all strictly convex, translation-invariant costs (\cref{def:translationinvariant}).\footnote{
If $c(x,y) = \phi(y-x)$), then $c_{xy} = -\phi''$. 
Provided $\phi$ is strictly convex ($\phi'' > 0$), this implies $c_{xy} < 0$ (strictly submodular) and $|c_{xy}| = c_{yy}$ (satisfying the non-expansiveness condition with equality).}
Consequently, standard specifications such as quadratic variation and Pseudo-Huber loss yield non-expansive policies.
The optimal policy thus acts as a \emph{damped smoother}, strictly interpolating between the \RR ($r'(p)=1$) and the \AR ($r'(p)=0$).

\medskip 
This theoretical result validates our numerical observation from the uniform exogenous environment, see \Cref{fig:optimal_map}: the optimal agent acts as a \emph{damped smoother}, strictly interpolating between pure Status Quo ($r'(p)=1$) and pure Anchoring ($r'(p)=0$).

\begin{figure}[ht]
\centering
\includegraphics[width=0.55\textwidth]{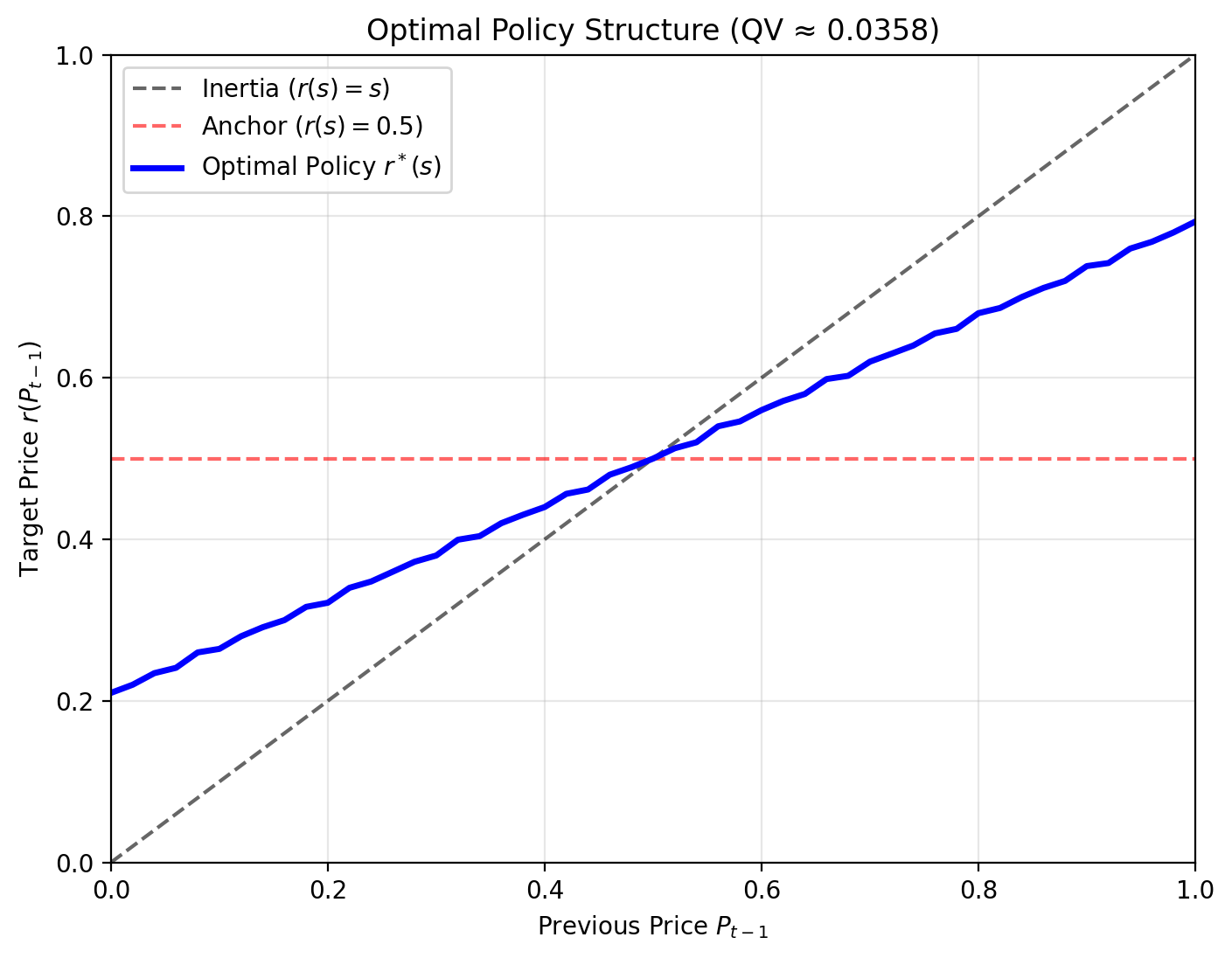}
\caption{\textbf{The Global Optimal Policy in the Uniform Environment for Quadratic Variation.} The solid blue line ($r^\ast(s)$) depicts an approximation of the optimal target function, solved via an Average Cost Bellman Equation. It interpolates between the \RR ($45^\circ$, dashed) and the \AR (horizontal), pulling the action toward $s=0.5$ with a slope of $\approx 0.88$. However, the efficiency gain is marginal compared to the \RR, which captures over \textbf{99\%} of the theoretical limit.}
\label{fig:optimal_map}
\end{figure}

\subsection{\textbf{Optimality of the \RR: The Random Walk Limit}}

In the previous section, we imposed the stability condition given in \cref{assum:stability} ($\mathbb{E}[|A_t|^p] < 1$). 
However, many financial applications focus on ``level-neutral'' environments --- such as random walks, where actions are fully persistent ($A_t = 1$ a.s. and \cref{assum:stability} is violated), and the action path has infinite variance.

\begin{assumption}[Random Walk Environment]
The transition parameter is $A_t \equiv 1$ almost surely. The feasibility set evolves additively: $I_t \;=\; P_{t-1} + J_t,$
where the base increments $(J_t)_{t \ge 1}$ are i.i.d. random intervals with finite moments.
\label{assum:randomwalk}
\end{assumption}

We are able to address this boundary case that is relevant in practice separately, and show that the \RR is uniquely optimal in this environment if adjustment costs are translation invariant (\cref{def:translationinvariant}).

\begin{theorem}[Optimality of the \RR]
\label{theorem:rwinertia}
In a Random Walk environment (\Cref{assum:randomwalk}), suppose that adjustment costs are translation invariant. Then, the  \RR  is the unique optimal policy.
\end{theorem}

\noindent\textbf{Proof Outline and Intuition.}
The formal proof (provided in \Cref{ssec:ProofTheorem2}) exploits the symmetry of the Random Walk environment. Unlike the mean-reverting case, where the agent must navigate toward a ``safe'' center to avoid future costs, the Random Walk environment is translation invariant.

\begin{enumerate}[leftmargin=*, labelindent=0pt]
    \item \textbf{Decoupling.}
    In a random walk environment, the distribution of future constraint increments is independent of the current absolute action level. Shifting the entire action path up or down does not change the tightness of future constraints. Consequently, there is no good region of the state space to drift toward, meaning the differential value function $h(p)$ becomes constant.

    \item \textbf{Collapse to Static Problems.}
    With the future value term constant, the intertemporal trade-off disappears. The infinite-horizon optimization decouples into a sequence of independent static problems. 
    At each step $t$, the decision maker minimizes the immediate adjustment cost $\phi(\Delta P_t)$.

    \item \textbf{Optimality of Projection.}
    For any strictly convex cost function minimized at zero, the optimal action is to choose the smallest possible increment. Geometrically, this is achieved by projecting the previous action $P_{t-1}$ onto the current feasible interval $I_t$. Thus, the \RR is the unique optimal policy. \hfill $\diamondsuit$
\end{enumerate}

\subsection{Minimizing Variance: Optimality of the Anchoring Rule}
\label{ssec:variance}

So far, we considered adjustment costs that depend on the deviation from the previous period. 
However, in many applications, the decision maker has a \emph{global objective}: maintaining stability around a specific level rather than simply smoothing changes.
This objective is considered, e.g., in macro-finance, where the 100-day moving average is a classic metric and in behavioral economics, where agents reject prices that drift too far from a customary reference point.\footnote{In behavioral settings, \cite*{Kahneman1986} establish that consumers perceive prices deviating from a ``reference transaction'' as unfair, effectively imposing a penalty on variance.} 

In this section, we extend our model to the objective of minimizing the \emph{variance} of the action path.
Mathematically, minimizing variance is equivalent to finding a policy $M$ and a static center $c$ that minimize the Mean Squared Error (MSE):
\[
\min_{M \in \mathcal M} \mathrm{Var}(P) \;=\; \min_{c \in \mathbb{R}} \left( \min_{M \in \mathcal M} \mathbb E_{\pi^M}\left[ (P_t - c)^2 \right] \right)\tag{Minimizing Variance}\label{eq:Var}
\]

The next theorem shows that for this objective, the optimal strategy is the \AR, thus dampening as found to be optimal in the previous results is no longer present.

\begin{theorem}[Optimality of Constant Anchoring]
\label{theorem:varanchor}
Consider the objective of minimizing the stationary variance of the process.
The uniquely optimal policy is the \AR. There exists a unique optimal target $z \in \mathbb{R}$ and the target function is constant, $r(p) \;\equiv\; z$.
The optimal action is always the projection of this fixed target onto the current feasible interval:
$P_t^* \;=\; \Pi_{I_t}(z)$.
\end{theorem}

This result provides a sharp contrast to the \RR found in \Cref{theorem:rwinertia}. While the \RR yields $r(p)=p$ (slope 1), Variance minimization yields the opposing policy $r(p)=c^\ast$ (slope 0).

\smallskip
\noindent\textbf{Proof Outline and Intuition.}
Variance measures global dispersion. 
History-dependent policies such as \RR introduce correlation, allowing actions to drift. 
The \AR forces actions to return to a fixed target immediately whenever feasible, eliminating drift and confining the distribution as tightly as possible.
The proof (provided in \Cref{ssec:ProofTheorem3}) decomposes the global variance problem:

\begin{enumerate}[leftmargin=*, labelindent=0pt]
    \item \textbf{Decomposition (MSE Equivalence).}
    We use \Cref{eq:Var}. 
    We fix an arbitrary center $c$ and solve for the policy $\pi_c$ that minimizes the expected squared distance to $c$.

    \item \textbf{The Static Target Problem.}
    For a fixed $c$, the cost function is $c(P_t) = (P_t - c)^2$. 
    The Bellman equation for this problem seeks an action $y \in I_t$ that balances the immediate distance to $c$ against future feasibility.
    We show that the strictly convex structure implies that the optimal policy is to project onto a constant target $z_c$ determined by the full dynamic program.
    Thus, the optimal policy is a projection $P_t = \Pi_{I_t}(z_c)$.

    \item \textbf{Global Optimization.}
    The problem then reduces to finding the optimal scalar target $c^\ast$ that minimizes the long-run average squared deviation. The global optimum is the policy that projects onto the corresponding optimal anchor $z_{c^\ast}$ in every period. \hfill $\diamondsuit$
\end{enumerate}

Crucially, under endogenous constraints, the optimal anchor $z_c^\ast$ generally differs from the process mean $c^\ast$.
To minimize global dispersion, the policy must anchor to a distinct value that compensates for the dynamic drift of the constraints.
However, this distinction affects only the location of the target, not the policy's structure: because the objective function remains state-independent, the optimal strategy is always to project a unique, constant target onto the current feasible set.

\bigskip
\textbf{Exogenous Constraints.} To provide a sharper characterization of the optimal anchor, we consider the benchmark case where feasibility constraints are exogenous (i.e., independent of past actions). 
The optimal policy projects directly onto an anchor that balances the marginal costs of upward and downward deviations.

\begin{theorem}[Characterization of the Optimal Anchoring Rule]
\label{theorem:optanchor}
Assume the evolution of the feasible interval is exogenous (i.e., $A_t = 0$). The variance-minimizing anchor $z^\ast$ is the unique solution to
\begin{equation}
\label{eq:anchor-balance}
\mathbb{E}\Big[ (L_t - z^\ast) \cdot \mathbf{1}_{\{L_t > z^\ast\}} \Big] \;+\; \mathbb{E}\Big[ (R_t - z^\ast) \cdot \mathbf{1}_{\{R_t < z^\ast\}} \Big] \;=\; 0. \tag{Balance Condition}
\end{equation}
\normalsize
Equivalently, the optimal anchor satisfies a self-consistency condition, meaning the target must coincide with the resulting stationary mean of the process: 
$z^\ast \;=\; \mathbb{E}_{\pi}\big[ \Pi_{I_t}(z^\ast) \big]$. 
If the joint distribution of $(L_t, R_t)$ is symmetric about some center $c$, the solution simplifies to $z^\ast = c$.
\end{theorem}

\noindent\textbf{Proof Outline and Intuition.}
The full proof is relegated to \Cref{ssec:ProofTheorem4}.
Minimizing the variance is equivalent to minimizing the expected squared displacement $\mathbb{E}[(\Pi_{I_t}(z) - z)^2]$. 
The first-order condition requires the net expected displacement to be zero. 
The term $(L_t - z)\mathbf{1}_{\{L_t > z\}}$ captures positive displacement when the lower boundary is active, while $(R_t - z)\mathbf{1}_{\{R_t < z\}}$ captures negative displacement when the upper boundary is active. \hfill $\diamondsuit$

\bigskip
In our uniform environment discussed in \cref{sec:introduction}, the balance condition is satisfied at $z^\ast = \frac{1}{2}$, confirming that the \AR anchored at $\frac{1}{2}$ is the unique variance-minimizing mechanism.
The \AR yields substantial gains over both the Midpoint Rule and the \RR.
Relative to the Midpoint Rule, it reduces variance by 50\%.
Relative to the \RR, it achieves a variance reduction of approximately 41\%.
For the analytical derivation of the resulting variance for each rule, see \cref{app:uniform-benchmark}. 

\vspace{-0.25cm}
\section{Applications}\label{sec:applications}

We next discuss several applications that directly fit into our general modeling framework.

\vspace{-0.25cm}
\subsection{Double Auction}\label{ex:doubleauction}

\emph{Double Auctions} --- often referred to as \emph{Call or Batch Auctions} --- are a central mechanism in economic theory and are ubiquitous in financial markets \citep{Fri93}. 
In theory, the double auction is appealing due to its simplicity and detail-freeness,\footnote{Referred to as `Wilson's doctrine', a mechanism ought to be detail-free and not rely on features of the agents' common knowledge \citep{Wil85}. The double auction is Wilson's leading example of a real-world mechanism with these features.} because it implements Walrasian exchange \citep{Wil85}, and because of its good incentive and efficiency properties \citep{Rus94,Sat02,Crip06,Per06,Az19}.

In practice, they are the standard mechanism for opening and closing major stock markets, and they are used to clear the market periodically for illiquid assets (\emph{Periodic Call Auction}), see, e.g., \cite{LSE_MIT201}. 
They have also been discussed as a general alternative for intraday trading (\emph{Frequent Batch Auctions}), as they could mitigate the \emph{High Frequency Arms Race} \citep*{budish2015high}.
Furthermore, repeated batch auctions are important mechanisms in \emph{Decentralized Finance} (DeFi).\footnote{Protocols such as \citeauthor{cow_protocol}, \citeauthor{injective_protocol}, and \citeauthor{sei_protocol} have adopted frequent Double Auctions as their core settlement layer, processing over \$500 million in daily volume.}
The terms of trade in a Double Auction are determined as exemplified by the NYSE regulations: 

\singlespacing
\small{
\begin{quote}
\textsc{\textbf{\cite[Rule 7.35C. Exchange-Facilitated Auctions]{NYSE_Rule_7.35C}}}
\emph{ 
All orders eligible to trade $[...]$ will be matched and traded at the Indicative Match Price (IMP), which means the best price at which the maximum volume of shares $[...]$ is tradable. 
If there are two or more prices at which the maximum volume of shares is tradable, the IMP will be the price closest to the Auction Reference Price (ARP), provided that the IMP will not be lower (higher) than the price of an order to buy (sell) $[...]$ that was eligible to participate in the applicable auction. 
If there are two prices at which the maximum volume of shares is tradable and both prices are equidistant to the ARP, the IMP will be the ARP. 
}
\end{quote}}
\onehalfspacing
\normalsize
This rule implements Walrasian equilibrium, and selects from the range of equilibrium prices according to a reference price rule \citep{Jantschgi25}.
Consider a finite market with a set of unit-demand buy orders  and sell orders for a homogeneous asset. 
For any price $P$, the \emph{Demand Correspondences} $\mathcal{D}(P)$
is the interval of quantities demanded at price $P$; the lower bound is given by the quantity demanded with bids strictly above $P$ and the upper bound is given by the quantity demanded with bids at or above $P$. The \emph{Supply Correspondences} $\mathcal{S}(P)$ is defined symmetrically. 
The set of market-clearing prices is the closed interval where these correspondences intersect: $I \;=\; \{P \in \mathbb{R} \mid \mathcal{D}(P) \cap \mathcal{S}(P) \neq \emptyset \}.$

What price to implement within the clearing interval $I$ is a design choice.
In repeated auctions, in each period $t=1,2,\ldots$, the auctioneer faces a new interval $I_t$ and has to choose a market clearing price $P_t\in I_t$.
The intervals naturally evolve stochastically due to exogenous shocks such as news releases or endogenous factors such as the information revealed by previous prices. 
Finally, if the market-clearing price in period $t$ factors in the formula that determines the market-clearing price in period $t+1$ an additional endogenous factor emerges.
This feedback loop implies that the auctioneer does not just select a price to clear the market today; they choose a control variable $P_t \in I_t$ that may influence the market in the next period, see \Cref{fig:auction_dynamics} below.

\begin{figure}[ht]
\centering
\setlength{\belowcaptionskip}{-10pt}
\begin{tikzpicture}[scale=0.625, >=stealth, font=\small]

    \tikzset{
         faded_demand/.style={blue, thick, opacity=0.25},
         faded_supply/.style={red, thick, opacity=0.25},
         interval_style/.style={green!60!black, line width=3pt},
         anchor_line/.style={dotted, black!70, thick}
     }

     \begin{scope}[local bounding box=period1]
         \draw[->, thick] (0,0) -- (7,0) node[right] {Quantity};
         \draw[->, thick] (0,0) -- (0,8.5) node[above] {Price};
        
         \node[below, font=\bfseries] at (3.5, -0.8) {Period $t-1$};

         \draw[faded_demand] (0,7.5) -- (1,7.5) -- (1,6.5) -- (2,6.5) -- (2,5.5) -- (3,5.5) -- (3,4.5) -- (4,4.5) -- (4,2.5) -- (5,2.5) -- (5,1.5) -- (6,1.5);
         \draw[faded_supply] (0,0.5) -- (1,0.5) -- (1,1.5) -- (2,1.5) -- (2,2.5) -- (3,2.5) -- (3,3.5) -- (4,3.5) -- (4,5.0) -- (5,5.0) -- (5,6.0) -- (6,6.0);

         \def\posA{5.5}
         \draw[dashed, gray!40] (4, 4.5) -- (\posA, 4.5);
         \draw[dashed, gray!40] (4, 3.5) -- (\posA, 3.5);
         \draw[interval_style] (\posA, 3.5) -- (\posA, 4.5);
        
         \node[right, green!60!black] at (\posA, 3.5) {$L_{t-1}$};
         \node[right, green!60!black] at (\posA, 4.5) {$R_{t-1}$};
        
         \node[circle, fill=black, inner sep=1.5pt, label={[left, xshift=1pt]$P_{t-1}$}] (Pt_dot) at (\posA, 4.45) {};
     \end{scope}

     \begin{scope}[shift={(14,0)}, local bounding box=period2]
         \draw[->, thick] (0,0) -- (7,0) node[right] {Quantity};
         \draw[->, thick] (0,0) -- (0,8.5) node[above] {Price};

         \node[below, font=\bfseries] at (3.5, -0.8) {Period $t$};

         \def\yshift{0.5}

         \begin{scope}[yshift=\yshift cm]
              \draw[faded_demand] (0,7.5) -- (1,7.5) -- (1,6.5) -- (2,6.5) -- (2,5.5) -- (3,5.5) -- (3,4.5) -- (4,4.5) -- (4,2.5) -- (5,2.5) -- (5,1.5) -- (6,1.5);
              \draw[faded_supply] (0,0.5) -- (1,0.5) -- (1,1.5) -- (2,1.5) -- (2,2.5) -- (3,2.5) -- (3,3.5) -- (4,3.5) -- (4,5.0) -- (5,5.0) -- (5,6.0) -- (6,6.0);
         \end{scope}

         \def\posB{5.5}
         \draw[dashed, gray!40] (4, 3.5+\yshift) -- (\posB, 3.5+\yshift);
         \draw[dashed, gray!40] (4, 4.5+\yshift) -- (\posB, 4.5+\yshift);
         \draw[interval_style] (\posB, 3.5+\yshift) -- (\posB, 4.5+\yshift);
        
         \node[right, green!60!black] at (\posB, 3.5+\yshift) {$L_{t}$};
         \node[right, green!60!black] at (\posB, 4.5+\yshift) {$R_{t}$};
        
         \coordinate (AnchorTarget) at (0, 4.5);
         \draw[anchor_line] (AnchorTarget) -- (\posB, 4.5);
        
         \node[left, xshift=-2pt] at (AnchorTarget) {$C_{t}(P_{t-1})$};
     \end{scope}

     \draw[->, very thick, black!60, dashed, bend left=50] 
         (Pt_dot) to node[midway, above=8pt, align=center, font=\footnotesize, black] {} (AnchorTarget);

\end{tikzpicture}
\caption{\textbf{Dynamics of the Market Clearing Interval in Repeated Double Auctions.} An example of repeated double auctions with seven buy and seven sell orders. The intersection of demand (blue) and supply (red) step functions defines a closed interval of market-clearing prices $I$. 
The market clearing-price chosen in the first period influences the distribution of demand and supply in the second period. The graphic illustrates the evolution of the market across two periods. In period $t-1$ (left), given realized demand and supply shocks, the clearing interval is $I_{t-1}=[L_{t-1}, R_{t-1}]$. A specific trade price $P_{t-1}$ is selected (here, the upper bound). In period $t$, trader valuations have changed, shifting the expected center of demand and supply.}
\label{fig:auction_dynamics}
\end{figure}
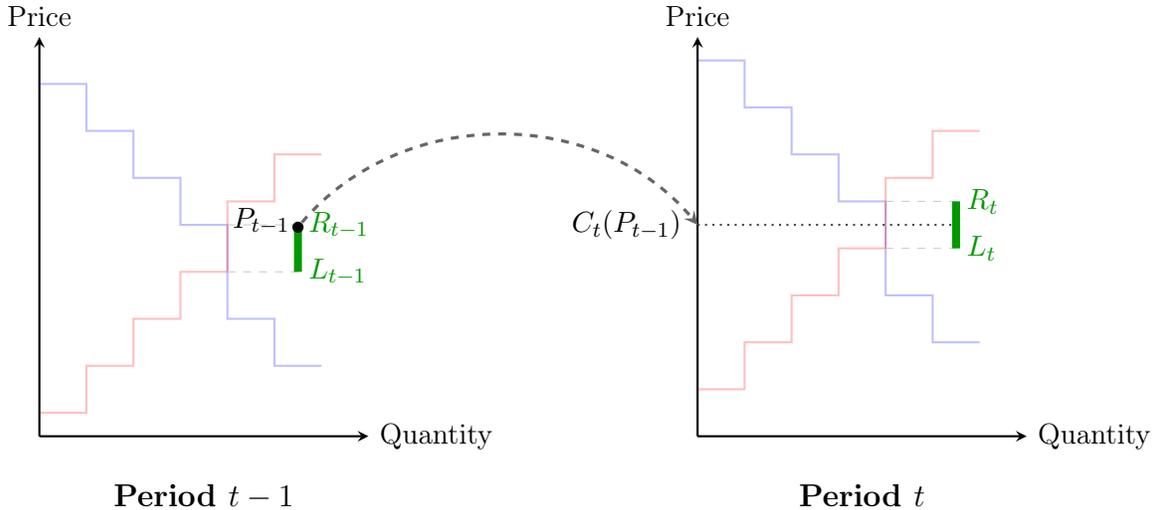

\newpage
Standard market microstructure models often assume random walk environments. Under these conditions,  \Cref{theorem:rwinertia} shows that the \RR is the globally optimal policy.
Strikingly, this theoretical optimum corresponds exactly to the practice of major stock exchanges, including the New York Stock Exchange and the London Stock Exchange.
In contrast, early academic proposals for Frequent Batch Auctions \citep{budish2015high} and current implementations in several DeFi protocols often utilize a \emph{Midpoint Rule} ($P_t = \frac{L_t + R_t}{2}$).\footnote{This standard solution is a special case of a \emph{$k$-Double Auction} \citep{Wil85,Rus94}, which selects a price based on a fixed relative position $k \in [0,1]$ between the endpoints: $P = (1-k) L + k R$.}
Our analysis reveals that such rules are strictly suboptimal for minimizing volatility. 
In our standard uniform example, adopting a Midpoint Rule instead of a \RR effectively doubles the volatility of the price process.

We provide a more detailed micro-foundation of our abstract model for Double Auctions in \Cref{ssec:microfoundation}. 
In \cref{app:incentiveinauctions}, we discuss the incentive properties of Double Auctions with Reference Prices. For a broad class of action distributions, they are Strategyproof in the Large \citep{Az19,Jantschgi25}.

\subsection{Over-the-Counter Markets (Bilateral Trade)}

While exchanges dominate equity trading, the vast majority of bonds, currencies, and derivatives trade in bilateral \emph{Over-the-Counter (OTC)} markets, see, e.g., \cite{duffie2005over,hugonnier2025economics}.
In a typical transaction, an institutional client negotiates a price with a dealer (a bank or market maker).
In this context, frequent or erratic fluctuations can damage trust or signal instability; creating a natural incentive to minimize adjustment costs. 
Thus, both dealers and clients often prefer to keep prices steady unless fundamental values shift significantly.

This setting can be modeled as a bilateral version of the Double Auction in \Cref{ex:doubleauction}. 
Consider a single buyer with valuation $v_t$ and a seller with cost $c_t$ trading via sealed bids. 
If the buyer's bid $b_t$ exceeds the seller's ask $s_t$, the set of market-clearing prices simplifies to the bargaining range: $I_t \;=\; [s_t, b_t].$
Trade occurs if $s \le b$, and the transaction price $P$ is determined by projecting the reference price $R$ onto the feasible interval $[s, b]$. Specifically:
\begin{small}
\[
P(s, b) \;=\; 
\begin{cases} 
s & \text{if } R < s \le b \quad (\text{Seller's ask binds}) \\
R & \text{if } s \le R \le b \quad (\text{Reference price prevails}) \\
b & \text{if } s \le b < R \quad (\text{Buyer's bid binds})
\end{cases} \tag{Reference Price in Bilateral Trade}
\]
\end{small}
In this subsection, we ask, how do reference rules perform in this bilateral trade setting, where the interval is endogenously given by the actions of two traders?

This rule introduces a sharp non-convexity in incentives, visualized in \Cref{fig:bidding_strategy}.
The reference price $R$ forces the buyer to choose between two distinct strategic roles:
\begin{enumerate}[leftmargin=*, labelindent=0pt]
    \item \textbf{Price Maker (Monopsony):} If the buyer bids below $R$, the price is determined by their bid ($P=b$). Like a classic monopsonist, they aggressively shade their bid to minimize costs.
    \item \textbf{Price Taker:} If the buyer bids above $R$, the price is projected to $R$. The buyer effectively faces a fixed price; since raising the bid further has no immediate impact on the payment, they switch to truthful bidding to maximize the probability of trade.
\end{enumerate}

This regime switch creates a discontinuity: at a critical valuation $\hat{v}$, best responses jump from a low, shaded bid to a high, truthful bid, leaving a strategic ``hole'' below $R$.
See \Cref{fig:bidding_strategy} below for the optimal bidding strategy against a seller with uniform action distribution.

\begin{figure}[h]
    \centering
    \includegraphics[width=0.55\textwidth]{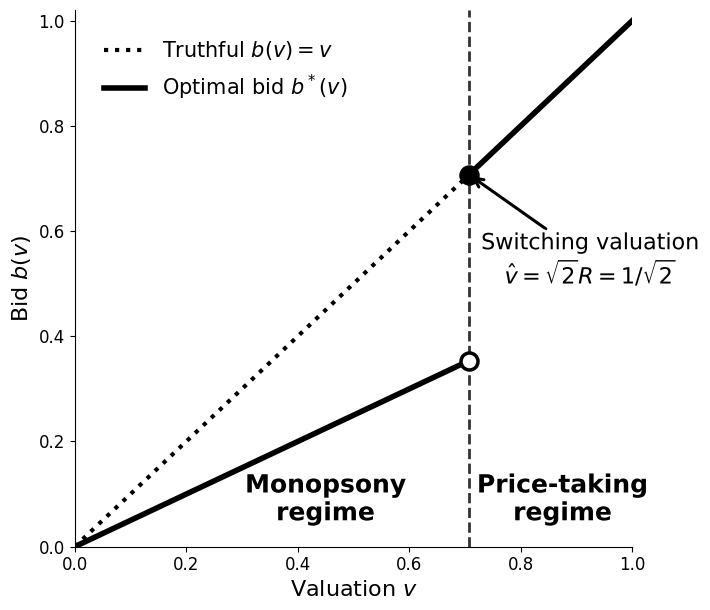}
    \caption{\textbf{Optimal Bidding Strategy under Uniform Distributions ($R=1/2$).}
    This example assumes costs and valuations are uniformly distributed on $[0,1]$ with a fixed reference price $R = 1/2$.
    The buyer compares the optimal monopsony utility $v^2/4$ (achieved by shading $b = v/2$) against the truthful price-taking utility $U_{stat}(v) = \frac{1}{2}(v^2 - 1/4)$.
    Equating these payoffs identifies a unique switching threshold at $\hat{v} = \sqrt{2}/2 \approx 0.707$.
    The figure plots the resulting optimal bid $b^*(v)$ against valuation $v$.
    Buyers below $\hat{v}$ aggressively shade their bids (reaching a maximum bid of $\approx 0.35$), while those just above $\hat{v}$ jump immediately to truthful bidding ($b=v$) to secure the reference price.
    This strategic discontinuity creates a distinct ``hole'' in the bid distribution between $0.35$ and $0.707$.}
    \label{fig:bidding_strategy}
\end{figure}

In \Cref{ssec:bilateral}, we extend this to a full equilibrium analysis. 
We show that this discontinuity causes continuous bidding equilibria to break down.
There is, however, a robust outcome: a \emph{Pooling Equilibrium} where all trade occurs exactly at the reference price $P=R$.

We thus contribute to the study of incentives in bilateral trade initiated by \cite{myerson1983efficient}, who established the impossibility of a mechanism that is simultaneously strategy-proof, efficient, and budget-balanced. The incentive properties of the \emph{$k$-Double Auction} (where $P = k s + (1-k) b$) are well-studied, with \cite{chatterjee1983bargaining} characterizing linear equilibria for the midpoint rule ($k=0.5$) and \cite{Sat89} extending this to general Bayesian Nash equilibria. 

Our analysis shows that the incentive to minimize volatility and the use of Reference Price changes the structure of Bayesian Nash equilibria.
By incentivizing traders to pool at $R$, the mechanism effectively degenerates into a \emph{Posted Price Mechanism} over a range of valuations. Crucially, this aligns with \cite{hagerty1987robust}, who showed that in settings requiring robust implementation (dominant strategies), posted prices are effectively the only viable option. 

\subsection{Resource Scheduling and Capacity Planning}

Adjustment costs have received a lot of attention in the theory of the firm, where production technology and capacity planning require taking them into account. The literature initiated by  \cite{Arrow51,scarf1960optimality} and most notably \cite{becker1983simple} study adjustments driven by changes in technology (for a survey, see \cite{becker2008capital}).  
This literature is not interested in exogenous constraints on the feasible decision set.  
In particular, our results on the optimality of reference rules has no forerunner in this literature. Instead, the main message of this literature is that the presence of adjustment costs smooths the path of technology adjustments. Larger technology shocks lead to smaller adjustments in the short-term than in the long-term; a force also present in our setting as established in \cref{theorem:non-expansive}.\footnote{In this line of work, the presence of adjustment costs has been shown to result in non-trivial dynamic optimization problems, and understanding the dynamic nature of adjustment costs helps explain various economic phenomena such as over- and under-investment, over- and under-staffing, over- and under-production, etc. Others, focus on the timing of irreversible investment decisions in the presence of uncertainty \citep{dixit-pindyck1994investment}.} 

We consider the general problem of capacity planning in environments where the ability to supply in the future is influenced by the output of the last period. 
Specifically, we study \emph{Just-in-Time (JIT) production}, where storage is strictly minimized either because of working-capital considerations or because the goods are perishable (such as fresh produce or daily print media). 

In these settings, the firm operates under constraints defining a range of viable options $I_t = [L_t, R_t]$. 
The upper bound $R_t$ may be determined by physical ramp-limits (e.g., thermal stress on turbines), available labor or raw materials, or maximum component flow. 
The lower bound $L_t$ may be set by minimum run-rates required to prevent idling, contractual delivery minimums, or safety baselines. 
In JIT systems, this range depends on staffing levels and machine configurations established in the previous period ($P_{t-1}$). 
Because the system lacks inventory buffers, expanding output requires immediate and costly scaling of inputs, while lowering it leaves committed resources idle. 
For perishable goods, since inventory expires immediately, the capacity to deliver tomorrow depends on distribution channels that are optimized for today's volume. 
If the firm delivered volume $P_{t-1}$ yesterday, its fleet and routing are calibrated for that scale; significant deviations require a reconfiguration of the logistics network. 
Modern infrastructure sectors like \emph{Energy Grids} and \emph{Cloud Computing (GPU)} represent a hybrid of these dynamics, combining instant perishability of the good with the physical inertial constraints of JIT manufacturing.\footnote{While an unused watt-hour or GPU cycle cannot be stored (perishability), the machinery producing it faces strict physical limitations on its rate of change (JIT inertia). In power generation, thermal plants face ramp rate limits to avoid mechanical stress. Similarly, in high-performance computing, clusters are configured for specific workloads; reallocating capacity deviates from the established state, incurring switching latency and reconfiguration downtime.} 

In all these settings, the feasible interval evolves stochastically due to exogenous shocks --- such as machinery failures, supply chain disruptions, or changing demand --- while also depending on the output from the last period. 
Minimizing adjustment costs, which arise from the friction of adapting production to these changes, is often a primary objective.

As mentioned above in our discussion of the literature initiated by\cite{becker1983simple}, our results contribute to a theoretical foundation for the practice of \emph{Production Smoothing}.
If the firm's primary objective is to minimize the costs associated with volatility (e.g. Quadratic Variation), our main result (\Cref{thm:optimal_policy_adjustment}) implies that the optimal strategy is a \RR. 
The decision maker should smooth production based on the last output and expectations of future shocks. 
If future constraints fully rely on today’s production (e.g., mirroring the random walk environment), then, according to \Cref{theorem:rwinertia}, the firm should maintain the previous production level $P_t = P_{t-1}$ as long as it remains feasible ($P_{t-1} \in I_t$). Adjustments are triggered only when exogenous shocks to supply or demand force the firm's hand (i.e., when the interval boundaries $L_t$ or $R_t$ drift away from $P_{t-1}$). 
This implies that a firm should effectively ignore small fluctuations in demand or capacity availability, absorbing them into the slack of the feasible interval. In contrast, a policy that adjusts the production to every change in demand or supply (e.g. in an attempt to center production within the feasible range — a Midpoint Rule) would be strictly suboptimal.

\subsection{Political Positioning with Endogenous Feasibility}
\label{ssec:political_positioning}

Political decision-making is inherently constrained by the need for credibility and consistency. 
The literature has long recognized that policy instability undermines effectiveness, whether through the time-inconsistency of optimal plans \citep{BarroGordon1983}, resistance to reform \citep{FernandezRodrik1991}, or the erosion of voter trust \citep{acemoglu2013economics}. 
Voters and coalitions punish erratic shifts, creating a zone of acceptance that anchors expectations \citep{alesina2019loss, Krugman1991}.
We sketch a reduced-form application of our framework to analyze how policymakers navigate these dynamic constraints.

In each period $t$, a decision maker (e.g., a party leader or government) chooses a one-dimensional policy position $P_t \in \mathbb{R}$. 
Their choice is constrained by political viability: at any time, institutional and electoral realities define a feasible interval $I_t = [L_t, R_t]$. 
Positions within this interval are safe, maintaining coalition stability and satisfying the median voter. 
In contrast, deviations beyond $L_t$ or $R_t$ generate discrete and irreversible losses, such as loss of office or breakdown of governing coalitions. Thus, boundary violations generate first-order losses, while interior adjustments have only local effects.
Crucially, feasibility is path-dependent. 
Voters and media anchor their expectations on the status quo. A policy that was considered radical yesterday may become the center of the \emph{Overton window} today if the previous policy moved in that direction. 
We model this evolution as in our baseline with $I_t \;=\; C_t(P_{t-1}) + J_t$ and  $C_t(P_{t-1}) = A_t P_{t-1} + B_t$.
Here, $A_t$ captures the persistence of political norms (how much the window shifts with the previous policy), while $B_t$ and $J_t$ represent exogenous shocks such as demographic shifts, economic crises, or opponent behavior.

Our results illuminate a dichotomy in political strategy, depending on the decision maker's objective. 
If the primary cost is ``flip-flopping'' (adjustment costs), the objective is to minimize changes in policy $\Delta P_t$ to appear steady and predictable. 
In this case, our analysis shows that the optimal strategy is that the policymaker anchors their current decision to their own past position ($P_{t-1}$), see \Cref{thm:optimal_policy_adjustment}. 
This rationalizes the behavior of centrist ``pragmatists'' who evolve their agenda slowly, ensuring that today's policy always serves as the viable reference point for tomorrow's window.
In contrast, a policymaker may seek to minimize variance or the distance to a specific, fixed ideological agenda$\theta$ (e.g., a party constitution or long-term mandate). 
By the same logic used in our variance minimization result (\Cref{theorem:varanchor}), the optimal policy is an \AR: the decision maker projects the ideal position $\theta$ onto the current feasible interval $I_t$. 
Thus, our model formally distinguishes the optimal behavior of the ``pragmatist'' (who minimizes adjustment costs by drifting) from the ``ideologue'' (who minimizes ideological deviation by returning to a fixed base).

\subsection{Other applications}
In addition to the examples above, a wide range of other economic settings fit our framework where a decision maker faces hard participation or feasibility constraints and incurs adjustment costs.
Examples include platform fee setting, where fees must remain within a band that sustains developer or merchant participation; wage and compensation policies, where firms maintain pay within retention-compatible ranges shaped by reference dependence; monetary and regulatory policy, where credibility and legal constraints induce target zones; and corporate risk management, where leverage and inventory choices must remain within covenant- or rating-safe regions.
In all these cases, current decisions shift future feasibility through expectations or balance-sheet effects, boundary violations generate discrete losses, and dynamic adjustment frictions dominate fine-tuning within the feasible region. Our framework captures this common structure in a unified stochastic control model.

\section{Outlook}

Dynamic adjustment costs with evolving constraints present many more important questions. We here discuss one natural 
 frontier for future research, namely the extension of our results to higher-dimensional action spaces. 
In many mechanism design environments, the set of feasible outcomes is not a simple interval, but a polytope in $\mathbb{R}^n$. 

A prominent example is the design of financial markets for portfolio trading. 
\cite{budish2023flow} proposes \emph{Flow Trading}, a mechanism where market-clearing prices are determined via batch auctions for arbitrary linear combinations of assets. 
In this setting, the set of market-clearing prices is typically a multidimensional polytope. 
They find that selecting a price arbitrarily from this feasible interval may result in high artificial volatility. 
To restore stability, they effectively implement a high-dimensional analogue of our \RR, anchoring the selection to the previous market-clearing price vector, and thus minimizing the $L^2$ distance between $P_t$ and $P_{t-1}$ subject to the equilibrium constraints.

This pattern extends to non-monetary allocation problems. 
In pseudo-markets for course allocation \citep{budish2011combinatorial, prendergast2022allocation, nguyen2025efficiency} or the allocation of food to food banks, the set of competitive equilibrium prices is typically non-unique. 
Similar considerations apply to exchange problems without money --- such as kidney exchange, shift swaps and resource reallocation \citep{echenique2019fairness, echenique2023balanced, jantschgi2025competitive}. 
Our analysis suggests that Reference Rules such as the ones discussed in this paper may perform well in these multidimensional settings as well. 

We are naturally also keen on empirical analyses of related phenomena in both the one-dimensional and multidimensional cases, especially from stock markets where such rules are implemented.

\bibliographystyle{ACM-Reference-Format}
\bibliography{Citation}

\newpage

\appendix

\section{Additional Results}\label{app:additional_results}

\subsection{A Uniform Benchmark Environment}
\label{app:uniform-benchmark}

In this section, we provide the calculations for the uniform benchmark discussed in the introduction.
We assume the feasible interval $I_t = [L_t, R_t]$ is formed by the order statistics of two i.i.d.\ uniform draws on $[0,1]$, such that $L_t = \min(X_t^{(1)}, X_t^{(2)})$ and $R_t = \max(X_t^{(1)}, X_t^{(2)})$. 

This simplified i.i.d.\ structure preserves the fundamental trade-off between inertia and anchoring while allowing us to derive closed-form solutions for the stationary distributions of the three benchmark policies.

\subsubsection{The Midpoint Rule.}

Consider the midpoint rule $P_t = (L_t + R_t)/2$. This policy is \emph{memoryless}: Consequently, the induced action process is i.i.d.
The action $P_t$ is the average of the minimum and maximum of two independent uniform variables $X_t^{(1)}, X_t^{(2)}$. 
$P_t$ follows a symmetric \emph{triangular distribution} on $[0,1]$ with mode at $1/2$. Its density function is:
\[
f_{\mathrm{mid}}(p) \;=\;
\begin{cases}
4p & \text{if } 0 \le p \le 1/2, \\
4(1-p) & \text{if } 1/2 < p \le 1. 
\end{cases}
\tag{Action Distribution - Midpoint}
\]
We compute the variance using the density. Since the mean is $1/2$, we have that
$\Var(\pi_{\mathrm{mid}}) \;=\; \int_0^1 (p - 1/2)^2 f_{\mathrm{mid}}(p) \, dp$.
By symmetry, this equals $2 \int_0^{1/2} (p - 1/2)^2 (4p) \, dp$, which yields:
\[
\Var(\pi_{\mathrm{mid}}) \;=\; \frac{1}{24} \;\approx\; 0.0417.
\tag{Variance - Midpoint}
\]
Since the policy is memoryless, consecutive actions are independent. 
The mean squared increment decomposes into the sum of variances: $\mathrm{QV}(\pi_{\mathrm{mid}}) \;=\; \mathbb{E}[(P_t - P_{t-1})^2] \;=\; \Var(P_t) + \Var(P_{t-1})$ and
\[
\mathrm{QV}(\pi_{\mathrm{mid}}) \;=\; \frac{1}{12} \;\approx\; 0.0833.
\tag{Quadratic Variation - Midpoint}
\]

\subsubsection{The Anchoring Rule at $\frac{1}{2}$.}

Consider the policy $P_t = \Pi_{I_t}(1/2)$ projecting the center of the state space onto the feasible interval. Like the Midpoint rule, this policy is \emph{memoryless} because the target is fixed.
The distribution of the action is a mixture of a point mass and a continuous density.
\begin{enumerate}
    \item \emph{The Atom at the Center:} The action equals the target $1/2$ whenever the interval covers it. 
    Since $L$ and $R$ are order statistics of two uniform draws, the probability they straddle the midpoint is $1 - (1/4 + 1/4) = 1/2$. 
    Thus, the policy hits its exact target 50\% of the time.
    
    \item \emph{The Tails:} If the interval is entirely to the left ($R_t < 1/2$), the action is constrained to the boundary $R_t$. 
    The density of $R_t$ is $2r$. 
    Thus, on the domain $[0, 1/2)$, the density is $2p$. 
    By symmetry, on $(1/2, 1]$, the density is $2(1-p)$.
\end{enumerate}
Summarizing, the stationary law is defined by the generalized density function:
\[
f_{\mathrm{anch}}(p) \;=\; 2\min(p, 1-p) \;+\; \frac{1}{2}\delta\left(p - \frac{1}{2}\right).
\tag{Action Distribution - Anchoring}
\]

\noindent We compute the variance by integrating the squared deviation against this mixed distribution. The atom at $1/2$ contributes zero to the variance. We integrate over the continuous tails: $\Var(\pi_{\mathrm{anch}}) \;=\; 2 \int_0^{1/2} (p - 1/2)^2 (2p) \, dp \;=\; 4 \int_0^{1/2} (p^3 - p^2 + \frac{p}{4}) \, dp.$
Solving the integral yields the total variance:
\[
\Var(\pi_{\mathrm{anch}}) \;=\; \frac{4}{192} \;=\; \frac{1}{48} \;\approx\; 0.0208.
\tag{Variance - Anchoring}
\]
Since the policy is memoryless, the expected squared increment is the sum of the variances:
\[
\mathrm{QV}(\pi_{\mathrm{anch}}) \;=\; \Var(P_t) + \Var(P_{t-1}) \;=\; \frac{1}{24} \;\approx\; 0.0417.
\tag{Quadratic Variation - Anchoring}
\]

\subsubsection{The Status Quo Reference Rule.}

Consider the policy, where the decision maker selects the feasibleaction, which is closest to the action from the previous period: $P_t = \Pi_{I_t}(P_{t-1})$.
This policy $\pi_{in}$ induces a persistent \emph{Markov chain}, which admits a unique stationary distribution.
We can solve the integral equation for the stationary density. 
The transition kernel has a singular component (a probability mass at the current state $P_{t-1}$) and continuous components where the action jumps to $L_t$ or $R_t$.
It turns out that the density is unimodal and symmetric around $1/2$. It admits the closed form:
\[
f_{in}(p) \;=\; \frac{2p(1-p)}{\bigl(1-2p+2p^2\bigr)^2}, \qquad p\in(0,1).
\tag{Action Distribution - Inertia}
\]
We calculate the variance by integrating the squared deviation from the mean ($1/2$) against the stationary density (via trigonometric substitution or partial fractions):
\[
\Var(\pi_{\mathrm{in}}) \;=\; \int_0^1 \left(p - \frac{1}{2}\right)^2 \frac{2p(1-p)}{(1-2p+2p^2)^2} \, dp \;=\; \frac{\pi - 3}{4} \;\approx\; 0.0354.
\tag{Variance - Inertia}
\]
For the quadratic variation, let $D(x) = \mathbb{E}[(P_t - x)^2 \mid P_{t-1}=x]$ be the expected squared increment given the current action $x$. 
A quick calculation shows that $D(x) \;=\; \frac{x^4}{6} \;+\; \frac{(1-x)^4}{6}.$\footnote{A jump occurs if the new interval $[L, R]$ does not contain $x$.
If $R < x$, the action jumps to $R$. This happens with probability density $f_R(r)=2r$ on $[0, x]$.
If $L > x$, the action jumps to $L$. This happens with probability density $f_L(\ell)=2(1-\ell)$ on $[x, 1]$.
The expected cost is the sum of these jumps:
$D(x) \;=\; \int_0^x (r-x)^2 (2r) \, dr \;+\; \int_x^1 (\ell-x)^2 (2(1-\ell)) \, d\ell =\; \frac{x^4}{6} \;+\; \frac{(1-x)^4}{6}.$}
The Quadratic Variation is the expectation of the conditional cost under the stationary density:
\[
\mathrm{QV}(\pi_{\mathrm{in}}) \;=\; \int_0^1 \frac{x^4 + (1-x)^4}{6} g(x) \, dx \;=\; \frac{\pi}{12} - \frac{2}{9} \;\approx\; 0.0396.
\tag{Quadratic Variation - Inertia}
\]

\begin{table}[h]
\centering 
\renewcommand{\arraystretch}{1.2}
\begin{tabular}{lcc}
\hline
\textbf{Mechanism} & \textbf{Quadratic Variation} & \textbf{Variance} \\ \hline
Midpoint 
& \textcolor{red}{\textbf{1/12} $\approx$ \textbf{0.0833}} & \textcolor{red}{\textbf{1/24} $\approx$ \textbf{0.0417}} \\
\AR at $\frac{1}{2}$ 
& \textcolor{orange}{\textbf{1/24} $\approx$ \textbf{0.0417}} & \textcolor{green!60!black}{\textbf{1/48} $\approx$ \textbf{0.0208}} \\
\RR 
& \textcolor{green!60!black}{\textbf{$\pi/12 - 2/9$} $\approx$ \textbf{0.0396}} & \textcolor{orange}{\textbf{$(\pi-3)/4$} $\approx$ \textbf{0.0354}} \\ \hline \\
\end{tabular}
\caption{\textbf{Performance Trade-off.} \textcolor{green!60!black}{\textbf{Green}}, \textcolor{orange}{\textbf{Orange}}, and \textcolor{red}{\textbf{Red}} text indicates best, intermediate, and worst performance, respectively. Among the three policies, the~\RR is optimal for minimizing Quadratic Variation, while the~\AR is optimal for minimizing Variance.
Crucially, both policies substantially dominate the Midpoint Rule.  The~\RR reduces the Quadratic Variation by over $50\%$ ($0.0396$ vs. $0.0833$), while ~\AR cuts the Variance by  $50\%$ ($1/48$ vs. $1/24$).Neither policy dominates globally, but there is an asymmetry:  While~\RR incurs $70\%$ higher variance than~\AR ($0.0354$ vs. $0.0208$),~\RR outperforms Anchoring on quadratic variation by only $5\%$ ($0.0396$ vs. $0.0417$).}
\vspace{-2em}
\label{tab:uniform-results}
\end{table}

\newpage

\begin{figure}[ht]
\centering
\includegraphics[width=0.8\textwidth]{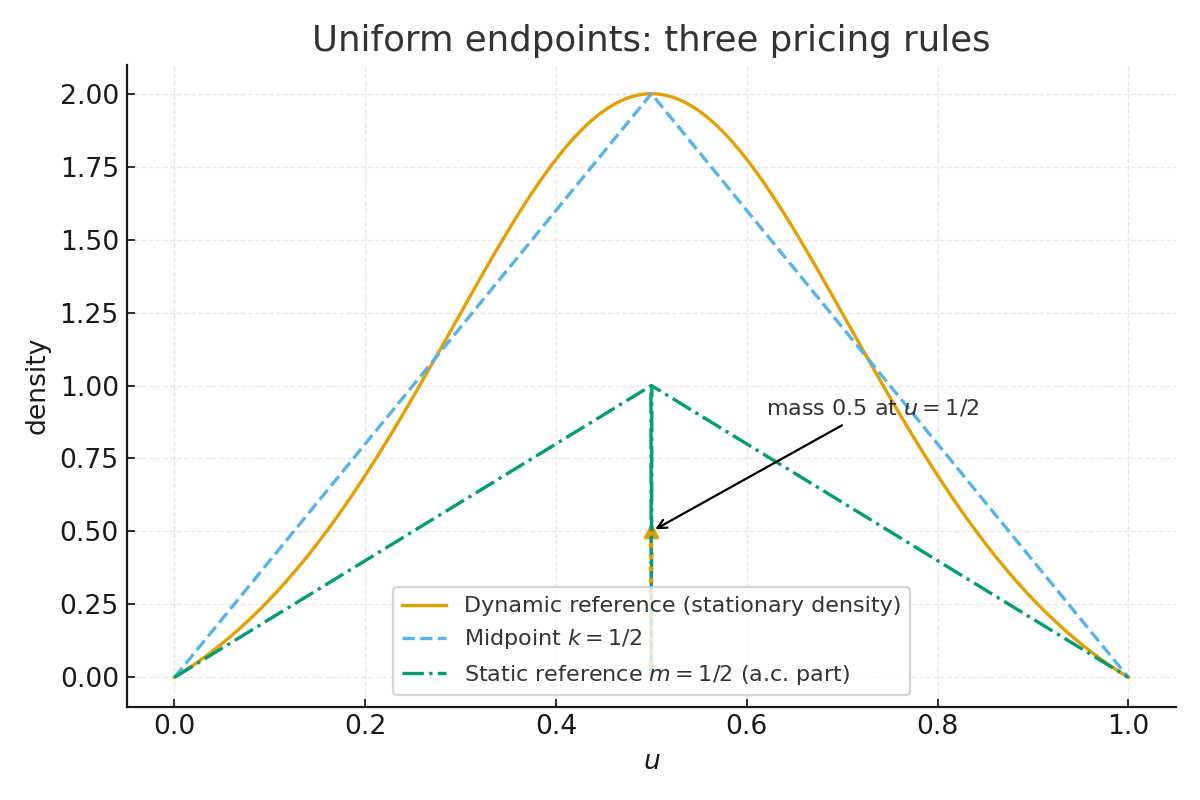}
\caption{\textbf{Stationary Distributions.} 
Midpoint ($f_{\rm mid}$, dashed) fluctuates mechanically with the interval. 
Anchoring ($f_{\rm anch}$, dotted) concentrates 50\% of the probability mass at the target (the vertical spike), minimizing global dispersion. 
Status Quo ($f_{\rm inertia}$, solid) is smootly dispersed, reflecting the action's tendency to drift within the feasible interval rather than correct toward the mean.}
\label{fig:uniform}
\end{figure}

\subsection{Pseudo-Code for Finding the Optimal Policy in \Cref{fig:optimal_map}}

To numerically find the optimal mechanism for minimizing quadratic variation in our uniform benchmark environment, we seek a solution to the \textbf{Average Cost Optimality Equation} (ACOE). The problem reduces to finding a scalar $\rho$ (the optimal average cost) and a relative value function $h(s)$ that satisfy the following Bellman equation for all states $s \in [0,1]$:

\begin{equation*}
\rho + h(s) \;=\; \min_{z \in [0,1]} \; \mathbb{E} \left[ \underbrace{\Big( \Pi_{I}(z) - s \Big)^2}_{\text{Immediate Adjustment Cost}} + \underbrace{h\Big( \Pi_{I}(z) \Big)}_{\text{Future Relative Cost}} \right].
\tag{ACOE}
\end{equation*}

Here, the expectation $\mathbb{E}$ is taken with respect to the random feasible interval $I = [L, R]$. Consistent with our model assumptions, the bounds are generated as order statistics of the uniform distribution: $L = \min(X_1, X_2)$ and $R = \max(X_1, X_2)$ where $X_1, X_2 \sim \mathcal{U}[0,1]$. 
The control variable $z$ represents the target reference action, and the system dynamics are governed by the projection $\Pi_{I}(z)$. 
The algorithm below solves this fixed-point equation using value iteration.

\begin{algorithm}[h]
\caption{Value Iteration for Quadratic Variation (ACOE)}
\label{alg:value_iteration}
\begin{algorithmic}[1]
\State \textbf{Initialize:}
\State \quad Discretize state space $\mathcal{S} = \{s_1, \dots, s_N\} \subset [0,1]$.
\State \quad Initialize relative value function $h(s) \leftarrow 0$ for all $s \in \mathcal{S}$.
\State \quad Set tolerance $\epsilon \leftarrow 10^{-9}$ and error $\delta \leftarrow \infty$.
\State \textbf{Pre-compute Noise:}
\State \quad Generate $M$ i.i.d.\ interval samples $[L_m, R_m]$ where $L_m = \min(X_m^{(1)}, X_m^{(2)})$, $R_m = \max(X_m^{(1)}, X_m^{(2)})$.

\While{$\delta > \epsilon$}
    \ForAll{$s \in \mathcal{S}$}
        \State \textit{// Find optimal unconstrained target $z$ for state $s$}
        \State Define objective function $J(z)$:
        \State \qquad $J(z) := \frac{1}{M} \sum_{m=1}^M \left[ \left(\Pi_{[L_m, R_m]}(z) - s\right)^2 + h(\Pi_{[L_m, R_m]}(z)) \right]$
        \State \qquad \textit{$\rhd$ Evaluate $h(\cdot)$ at off-grid points via linear interpolation over $\mathcal{S}$}
        \State $z^\ast(s) \leftarrow \argmin_{z \in [0,1]} J(z)$
        \State $h_{new}(s) \leftarrow J(z^\ast(s))$
    \EndFor
    
    \State \textit{// Normalize to prevent drift and check convergence}
    \State $\rho \leftarrow \min_{s} h_{new}(s)$ \Comment{Estimate average cost}
    \State $h_{norm}(s) \leftarrow h_{new}(s) - \rho$
    \State $\delta \leftarrow \max_{s} |h_{norm}(s) - h(s)|$ \Comment{Check convergence of relative values}
    \State $h(s) \leftarrow h_{norm}(s)$
\EndWhile

\State \Return Optimal Average Cost $\rho \approx \mathrm{QV}^\ast$, Optimal Policy $r^\ast(s) = z^\ast(s)$.
\end{algorithmic}
\end{algorithm}
\text{}
\normalsize

\newpage

\subsection{Embedding Stock Auctions in the General Framework}
\label{ssec:microfoundation}

Here, we provide the formal derivation connecting the repeated double auction market to the stochastic affine recursion $I_t = C_t(P_{t-1}) + J_t$ used in our general model in \Cref{sec:model}.

\smallskip
\textbf{The Auction Environment.}
Consider a market that operates over discrete periods $t=1, 2, \dots$. In each period $t$, a cohort of $N$ buyers and $M$ sellers enters the market to trade a single unit of a homogeneous asset. 
Each buyer $i$ has a private valuation $v_{i,t}$ (maximum willingness to pay), and each seller $j$ has a private cost $c_{j,t}$ (minimum willingness to accept). 
Traders submit limit orders---bids $b_{i,t}$ and asks $s_{j,t}$---to the exchange.
The market clears via a Double Auction mechanism. 
As established in the main text, the set of market-clearing prices is a closed interval $I_t = [L_t, R_t]$, the boundaries of which are determined by the order statistics of the submitted limit orders.

\smallskip
\textbf{Valuation Dynamics.}
We model the evolution of valuations as a process driven by two components: a common information process (anchored to past prices) and idiosyncratic heterogeneity.
Let $G$ and $F$ be fixed cumulative distribution functions with bounded support. 
We assume that in every period $t$, the idiosyncratic components of valuations, denoted by $\epsilon_{i,t}$ for buyers and $\delta_{j,t}$ for sellers, are drawn i.i.d. from $G$ and $F$ respectively.
The actual valuations are generated by shifting these base draws by a time-varying common component $Z_t$: $v_{i,t} = Z_t + \epsilon_{i,t}$ and $c_{j,t} = Z_t + \delta_{j,t}$.
The common component $Z_t$ represents the consensus fundamental value of the asset. 
We assume this fundamental value follows an AR(1) process anchored to the most recent public price signal $P_{t-1}$, subject to aggregate shocks: $Z_t = A_t P_{t-1} + B_t$.
$A_t$ captures the persistence of the price signal (market memory), and $B_t$ represents exogenous information shocks (news). We assume $(A_t, B_t)$ are i.i.d.

\smallskip
\textbf{Translation Invariance.}
We do not impose a specific solution concept on the auction. 
Instead, we rely on the mild assumption that whatever bidding strategies traders employ, they satisfy \emph{Translation Invariance}. 
This implies that if the underlying valuations are shifted by a constant $C$, the submitted bids are shifted by the same amount. 
Intuitively, this condition is robust because the profitability of a trade depends solely on the surplus margin $(v - P)$; adding a constant to both the valuation and the price level leaves the relative strategic incentives unchanged.

\smallskip
\textbf{Evolution of the Market-Clearing Interval.}
Let $J_t = [u_t, v_t]$ denote the clearing interval derived solely from the base valuations $(\epsilon_{i,t}, \delta_{j,t})$. 
By Translation Invariance, the actual bids—and consequently the resulting clearing boundaries—are simply the base values shifted by the anchor $Z_t$. 
Thus, $[L_t, R_t] = Z_t + J_t$. 
Substituting the definition of $Z_t$, we confirm that the feasible interval evolves exactly according to our general stochastic model:
\[
    I_t = C_t(P_{t-1}) + J_t,
\]
where $C_t(p) = A_t p + B_t$ is the stochastic anchor and $J_t$ represents the exogenous base width.

\subsection{Incentive Properties of Stock Auctions}\label{app:incentiveinauctions}

In this subsection, we provide the argument that Reference Rules in Double Auctions have good incentive properties.

\smallskip
\textbf{The Set-Up.} We assume that traders have quasi-linear preferences. 
Each trader $i$ has a private valuation (type) $v_i \ge 0$ for the asset. 
The utility of a trader depends on their valuation, the allocation, and the transfer. 
Specifically, if a buyer trades at price $P$, their utility is $v_i - P$; if a seller trades at price $P$, their utility is $P - v_i$. Traders who do not trade receive a utility of 0. 
Formally, given a strategy profile $q = (q_i, q_{-i})$ resulting in a market price $P(q)$ and an allocation where $i$ trades (denoted by $\mathbf{1}_i^{\text{trade}}$), the utility is:
\[
u_i(v_i, q) \;=\; (v_i - P(q)) \cdot \mathbf{1}_i^{\text{trade}} \quad (\text{for buyers}), \qquad
u_i(v_i, q) \;=\; (P(q) - v_i) \cdot \mathbf{1}_i^{\text{trade}} \quad (\text{for sellers}).
\]
Traders submit limit orders $q_i$ to maximize their expected utility.
In finite markets, Double Auctions are generally not strategy-proof. Intuitively, the market-clearing interval $[L, R]$ is pinned down by the bids of the marginal traders. 
A single trader can thus influence this interval by misreporting their valuation, thereby manipulating the final price in their favor.

\smallskip
\textbf{Asymptotic Incentive Compatibility.} 
However, as the market grows large, a single trader's influence on the clearing interval vanishes. 
In \cite{Jantschgi25}, we show that any price selection rule, and hence also reference price rules, become \emph{Asymptotically Incentive Compatible}.

\begin{proposition}[Asymptotic Incentive Compatibility, \cite{Jantschgi25}]\label{theorem:AIC}
Let trader $i$'s beliefs over other traders' quotes satisfy a random sampling regularity with marginal distributions that induce a unique market-clearing price in the limit. Then, truthful reporting is approximately optimal in sufficiently large Double Auctions with a Reference Price Rule. Formally, for any $\epsilon > 0$, in sufficiently large markets
\[
\mathbb{E}[u_i(v_i, v_i, q_{-i})] \;\ge\; \sup_{q_i} \mathbb{E}[u_i(v_i, q_i, q_{-i})] - \epsilon.
\]
\end{proposition}

While this theorem applies to any pricing rule, the Reference Price mechanism $P = \Pi_{[L, R]}(m)$ offers a strong intuitive defense against manipulation. 
Its incentive properties can be decomposed into three regimes:
\begin{enumerate}
    \item \textbf{Inside the Interval ($L < m < R$):} The price is set exactly at $m$. In this regime, the clearing price is locally independent of the marginal traders' bids. As long as the interval boundaries do not cross $m_t$, no trader can influence the price.
    \item \textbf{Below the Interval ($m \le L$):} The mechanism projects to the lower bound $L$. This effectively replicates a \emph{Seller's Ask Double Auction}, which is strategy-proof for buyers.
    \item \textbf{Above the Interval ($m \ge R$):} The mechanism projects to the upper bound $R$. This replicates a \emph{Buyer's Bid Double Auction}, which is strategy-proof for sellers.
\end{enumerate}
Unlike a fixed weighted-average rule (e.g., $P = \frac{L+R}{2}$), where every shift in the boundary $L$ or $R$ immediately shifts the price, the Reference Price rule ``disconnects'' the price from the boundaries whenever the reference is feasible.

\subsection{Incentives in Bilateral Trade}\label{ssec:bilateral}

We investigate how reference rules perform in a bilateral trade setting, where the interval is endogenously given by the actions of two traders.

\medskip
\textbf{Set-Up.} Consider a single-period trading game between a Seller and a Buyer.
The Seller has a cost $C$ drawn from a distribution $G$ with density $g$ on $[0,1]$. The Buyer has a valuation $V$ drawn from a distribution $F$ on $[0,1]$.
Agents submit limit orders: the Seller submits an ask $s$, and the Buyer submits a bid $b$.
We assume the distributions are \emph{regular} in the standard sense of mechanism design. Specifically, we assume the virtual valuation function $\Phi(v) = v - \frac{1-F(v)}{f(v)}$ and the virtual cost function $J(c) = c + \frac{G(c)}{g(c)}$ are strictly increasing.

Trade is governed by the \emph{Reference Price Rule} given an exogenous reference price $R \in (0,1)$ which represents a status quo, historical norm, or previous contract price.
Trade occurs if $s \le b$, and the transaction price $P$ is determined by projecting the reference price $R$ onto the feasible interval $[s, b]$. Specifically:
\[
P(s, b) \;=\; \Pi_{[s,b]}(R) \;=\; 
\begin{cases} 
s & \text{if } R < s \le b \quad (\text{Seller's constraint binds}) \\
R & \text{if } s \le R \le b \quad (\text{Reference price prevails}) \\
b & \text{if } s \le b < R \quad (\text{Buyer's constraint binds})
\end{cases}
\]

\subsubsection{\textbf{Best Responses}}

We begin by analyzing a single trader's optimization problem, taking the distribution of the counterparty's action as given. 
This best-response analysis serves as a fundamental building block toward the full equilibrium analysis, but it also holds independent economic relevance for settings where a single strategic trader faces a non-strategic counterpart.
We model this as a strategic buyer with valuation $v \sim F$ facing a passive seller with cost $c \sim G$. The seller truthfully asks their cost $s=c$. The buyer chooses a bid $b$ to maximize expected profit: $U(b; v) \;=\; \mathbb{E}_c \left[ (v - P(c, b)) \cdot \mathbb{I}(C \le b) \right]$.

The Reference Price Rule, $P(c, b) = \Pi_{[c, b]}(R)$, induces a sharp discontinuity, generating two distinct strategic regimes:

\begin{enumerate}
    \item \textbf{Monopsony Strategy ($b < R$):} The buyer bids below the reference price. In this regime, the feasible interval $[c, b]$ lies entirely below $R$, so the projection rule sets the price at the upper bound, $P=b$. The buyer acts as a standard monopsonist facing the supply curve $G(\cdot)$, maximizing $(v - b)G(b)$. The optimal bid $b^{mon}(v)$ is determined by the standard first-order condition involving the inverse hazard rate of the supply distribution.
    
    \item \textbf{Price-Taking Strategy ($b \ge R$):} The buyer bids at or above the reference price. Here, the projection rule enforces the reference price $P=R$ whenever the seller's cost is low enough ($c \le R$), and allows the price to rise to cost ($P=c$) only for marginal trades ($R < c \le b$). Conditional on entering this regime, the optimal tactic is to bid truthfully ($b=v$) to maximize the probability of trade, as the bid itself no longer sets the price for infra-marginal units. The value of this strategy is $U_{stat}(v) = (v - R)G(R) + \int_R^v (v - c)g(c)\, dc$.
\end{enumerate}

\begin{proposition}[Best Response Structure]
Assume the seller's distribution is regular, such that the virtual cost function $J(c) = c + \frac{G(c)}{g(c)}$ is strictly increasing. There exists a unique threshold valuation $\hat{v} > R$ such that the Buyer's optimal bidding strategy is characterized by a discontinuous regime switch:
\[
b^*(v) \;=\; 
\begin{cases} 
b^{mon}(v) & \text{if } v < \hat{v} \quad \text{(Monopsony Regime)} \\
v & \text{if } v \ge \hat{v} \quad \text{(Price-Taking Regime)}
\end{cases}
\]
\end{proposition}
The threshold $\hat{v}$ is the valuation where the buyer is indifferent between the high-margin, low trade probability monopsony strategy and the low-margin, high trade proability Price-Taking strategy.

\subsubsection{Bayesian Nash Equilibria}

We now extend the analysis to a full equilibrium setting, where both the buyer and seller act strategically based on private information. 
Here the feasible interval $[s, b]$ is fully endogenous, capturing scenarios where both parties hold power over the terms of trade.

Strategies $(s^*, b^*)$ form a Bayesian Nash Equilibrium if, for almost all $v, c$, they maximize expected surplus for both the buyer and seller:
\begin{align*}
b^*(v) &\in \arg\max_{b} \mathbb{E}_C [ (v - P(s^*(C), b)) \mathbb{I}(s^*(C) \le b) ], \\
s^*(c) &\in \arg\max_{s} \mathbb{E}_V [ (P(s, b^*(V)) - c) \mathbb{I}(s \le b^*(V)) ].
\end{align*}

We first show that \emph{non-degenerate equilibria} do not exist, which are continuous, monotone and with positive trade probability at the reference price.
This is in stark contrast to standard pricing rules, such as the $k$-Double Auction, for which a continuum of continuous and strictly increasing Bayesian Nash Equilibria exists \cite{chatterjee1983bargaining,Sat89}.
However, there exists a \emph{Pooling Equilibrium}, where trade only happens at the reference price.

\begin{proposition}[Equilibrium Characterization]\label{proposition:BNE}
Assume valuations $v$ and costs $c$ are independently distributed on $[0,1]$ with continuous, strictly positive densities.
\begin{itemize}
    \item \textbf{Impossibility of Non-Degenerate BNE:} There exists no Bayesian Nash Equilibrium with continuous, strictly increasing strategies such that the transaction price equals $R$ with strictly positive probability. 
    \item \textbf{Existence of Pooling Equilibrium:} There exists a Bayesian Nash Equilibrium where all trade occurs exactly at the reference price $P=R$. The strategies are:
    \begin{itemize}
        \item \textbf{Buyers:} If $v \ge R$, bid $b(v) = R$; otherwise bid $b(v) = 0$.
        \item \textbf{Sellers:} If $c \le R$, ask $s(c) = R$; otherwise ask $s(c) = 1$.
    \end{itemize}
\end{itemize}
\end{proposition}

The result stems from a sharp regime switch at the reference price $R$.
A buyer bidding below $R$ acts as a \emph{Price Maker} ($P=b$): raising the bid increases the transaction price, necessitating aggressive bid shading ($b \ll v$).
Conversely, a buyer bidding effectively above $R$ acts as a \emph{Price Taker} ($P=R$): raising the bid incurs zero marginal cost on the price, incentivizing truthful bidding ($b \approx v$).

This creates a strategic discontinuity: the marginal buyer indifferent between these regimes jumps from a low, shaded bid to a high, truthful bid, leaving a ``hole'' in the bid schedule that precludes continuous trade at $R$.

\medskip
\textbf{Pooling Equilibrium $\approx$ Posted Price Mechanism.} 
In the resulting Pooling Equilibrium, the market effectively degenerates into a \emph{posted price mechanism} at $R$.
Trade occurs if and only if $v \ge R \ge c$.
The problem of maximizing ex-ante welfare then reduces to selecting the optimal posted price. 
The welfare-maximizing posted price $R^*$ is the unique price that clears the market in expectation (the Walrasian price), satisfying $1 - F(R^*) = G(R^*)$. 
This aligns with the optimal \emph{robust} mechanism.
As shown by \cite{hagerty1987robust}, if one requires a mechanism to be dominant-strategy incentive compatible, posted-price mechanisms are the only solution.

\medskip

\textbf{The Cost of Stability.}
It is instructive to benchmark this result against the theoretical upper bound of mechanism design.
\cite{myerson1983efficient} established that no budget-balanced mechanism can achieve first-best efficiency (trading whenever $v \ge c$).
The second-best optimal mechanism maximizes surplus by enforcing a trading wedge $v \ge c + \alpha(v,c)$, excluding trades with thin margins.
Crucially, for the standard uniform case, the linear equilibrium of the simple Double Auction ($k=0.5$) exactly implements this theoretical second-best bound \citep{chatterjee1983bargaining}.
In comparison, the reference price rule with the pooling equilibrium leads to trade iff $v \geq 0.5 \geq c$.
This allows for a direct comparison between the optimal standard mechanism (Linear Equilibrium) and our Reference Rule.
As shown in \Cref{tab:welfare_comparison}, the Reference Rule captures approximately $89\%$ of the theoretical maximum surplus ($0.125$ vs $0.141$).
The table highlights the cost of stability: the Reference Rule sacrifices $\approx 11\%$ of potential welfare to achieve zero price volatility.

\begin{table}[h]
\centering
\renewcommand{\arraystretch}{1.2}
\begin{tabular}{lcc}
\hline
\textbf{Mechanism} & \textbf{Total Welfare} \\ \hline
First Best (Efficient) & $1/6 \approx 0.167$ \\
Myerson-Satterthwaite (Second Best)   &  $9/64 \approx 0.141$ \\
Linear Equilibrium in the $\frac{1}{2}-DA$  & $\mathbf{9/64 \approx 0.141}$ \\
Optimal Reference Price (Posted Price) & $1/8 = 0.125$ \\ \hline \\
\end{tabular}
\caption{\textbf{Welfare Comparison under Uniform Distributions.}
The \textit{First Best} represents efficient trade (impossible under budget balance).
The \textit{Myerson-Satterthwaite} bound represents the maximum surplus achievable under incentive compatibility.
For uniform distributions, the standard Linear Equilibrium of the $0.5$-Double Auction \citep{chatterjee1983bargaining} exactly achieves this second-best bound.
The \textit{Optimal Reference Price} mechanism ($R^*=0.5$) captures approximately $89\%$ ($0.125$ vs $0.141$) of the theoretical optimum.}
\label{tab:welfare_comparison}
\end{table}

\begin{figure}[t!]
    \centering
    \includegraphics[width=0.9\textwidth]{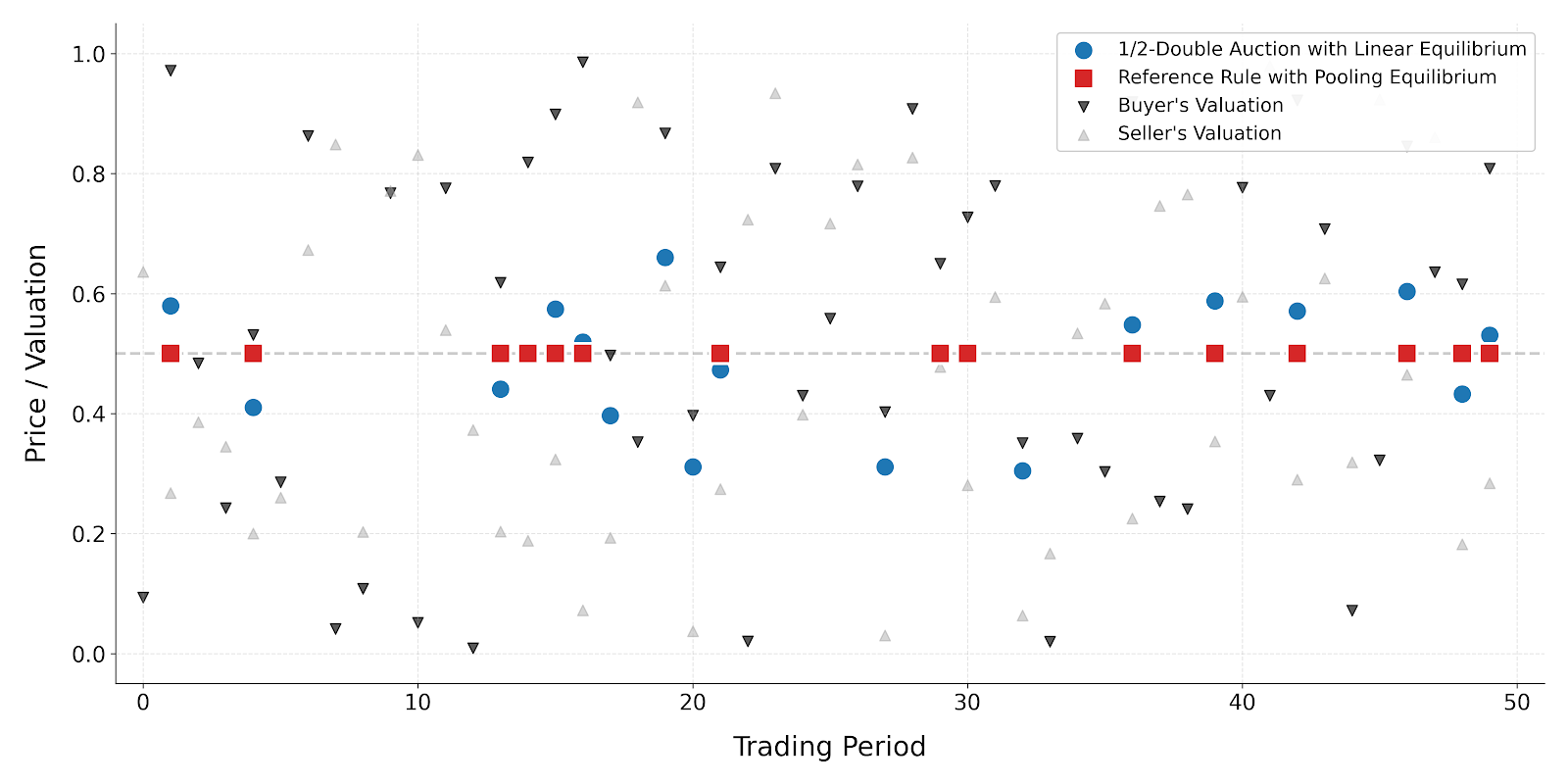}
    \caption{\textbf{Price Dynamics.} 
    This figure simulates $T=50$ periods of bilateral trade with uniformly distributed private values. In each period, a buyer valuation (black $\blacktriangledown$) and seller cost (gray $\blacktriangle$) are drawn uniformly from $[0,1]$. 
    \textbf{Blue circles} depict prices in the standard \textbf{1/2-Double Auction} (Linear Equilibrium), which achieves the theoretical second-best efficiency but exhibits high price volatility ($P_t$ fluctuates with valuations).
    \textbf{Red squares} depict prices under the \textbf{Reference Rule} (Pooling Equilibrium), which achieves perfect price stability at $R=0.5$.
    The simulation illustrates the trade-off: the Reference Rule successfully stabilizes the price but fails to realize ``off-center'' trades (e.g., high-value buyers matching with high-cost sellers) that the Double Auction successfully captures.}
    \label{fig:price_dynamics}
\end{figure}

\newpage

\subsection{Existence of Stationary Distribution}\label{ssec:stability}

In this subsection, we prove that \cref{assum:transitiondynamics} implies that for any admissible policy, the action path is well-behaved.
Moreover, if the policy is a stationary Markov policy, it induces at least one invariant measure.

\begin{proposition}[Universal Stability and $L^p$-Boundedness]
\label{proposition:universal_stability}
Under \cref{assum:transitiondynamics}  (specifically $\mathbb{E}[|A_t|^p] < 1$ for some $p \ge 2$), the action system is globally dissipative. Specifically:
\begin{enumerate}
\item For any admissible policy $\pi$, the $p$-th moments of the action process are uniformly bounded over time. That is, if $\mathbb{E}[|P_0|^p] < \infty$, then:
$$\sup_{t \ge 0} \mathbb{E}^\pi[|P_t|^p] \;<\; \infty.$$
\item If the policy $\pi$ is a stationary Markov policy, it induces at least one invariant probability measure $\psi_\pi$ with finite $p$-th moments. 
\end{enumerate}
\end{proposition}

The proof is relegated to \Cref{ssec:ProofAuxiliary}.

\newpage
\section{Proofs}

\subsection{Proof of \Cref{thm:optimal_policy_adjustment}}\label{ssec:ProofTheorem1}

\begin{proof}
\noindent\textbf{Step 1: Well-Posedness and Weighted-Norm Stability.}
The state space $\mathsf S = \mathbb{R}$ is non-compact, and the one-period adjustment cost function $c(x,y)$ is unbounded (\cref{assum:transitiondynamics}). Consequently, the value function may not belong to the space of bounded measurable functions. 
To establish the existence and uniqueness of the value function, we employ the \emph{weighted norm} framework for Markov Decision Processes (see, e.g., \cite{HernandezLermaLasserre1999}).

\smallskip
\textit{1.1 The Discounted Problem Setup.}
We consider the infinite-horizon $\beta$-discounted problem with discount factor $\beta \in (0, 1)$. The objective is to compute the value function:
\[
V_\beta(x) \;:=\; \inf_{\pi} \mathbb{E}^\pi \left[ \sum_{t=1}^\infty \beta^{t-1} c(P_{t-1}, P_t) \;\bigg|\; P_0 = x \right],
\]
where the infimum is taken over all admissible stationary Markov policies $\pi$.
Since the cost terms $c(P_{t-1}, P_t)$ can be arbitrarily large for large $x$ (due to polynomial growth), we must identify a suitable function space in which $V_\beta$ is well-defined and finite.

\smallskip
\textit{1.2 The Weighted Banach Space.}
We introduce the Lyapunov weight function $w: \mathbb{R} \to [1, \infty)$ defined by $w(x) = 1 + |x|^p$, where $p \ge 1$ is the growth order defined in \cref{assum:transitiondynamics}.
Let $\mathcal{B}_w$ denote the vector space of measurable functions $v: \mathbb{R} \to \mathbb{R}$ satisfying the growth condition:
\[
\|v\|_w \;:=\; \sup_{x \in \mathbb{R}} \frac{|v(x)|}{w(x)} \;<\; \infty.
\]
Equipped with the norm $\|\cdot\|_w$, $\mathcal{B}_w$ is a Banach space. We seek a solution $V_\beta$ within this space.

\smallskip
\textit{1.3 Lyapunov Stability of the Transition Kernel.}
A sufficient condition for the optimal value function to belong to $\mathcal{B}_w$ is the existence of at least one stationary policy that stabilizes the system's energy ($p$-th moment).
We construct a ``test policy'' $\pi_0$ for this purpose. Let $\pi_0$ be the midpoint policy:
\[
\pi_0(x, I_t) \;:=\; \frac{L_t + R_t}{2} \;=\; A_t x + B_t + \bar{J}_t,
\]
where $\bar{J}_t = (U_t+V_t)/2$ is the center of the noise interval.
We analyze the drift of the weight function $w(P_t)$ under this policy. By Minkowski's inequality, the $p$-th moment of the next state, conditional on $x$, satisfies:
\begin{align*}
\left(\mathbb{E}^{\pi_0}[|P_t|^p \mid P_{t-1}=x]\right)^{1/p}
\;&=\; \left(\mathbb{E}\big[ |A_t x + (B_t + \bar{J}_t)|^p \big]\right)^{1/p} \\
&\le\; \left(\mathbb{E}[|A_t x|^p]\right)^{1/p} + \left(\mathbb{E}[|B_t + \bar{J}_t|^p]\right)^{1/p} \\
&=\; (\mathbb{E}[|A_t|^p])^{1/p} |x| + K_0,
\end{align*}
where $K_0 = (\mathbb{E}[|B_t+\bar{J}_t|^p])^{1/p}$ is a constant depending on the noise moments.
Let $\rho_p := \mathbb{E}[|A_t|^p]$. By \cref{assum:stability}, $\rho_p < 1$. Raising both sides to the power $p$:
\[
\mathbb{E}^{\pi_0}[|P_t|^p \mid x] \;\le\; (\rho_p^{1/p} |x| + K_0)^p.
\]
For any $\varepsilon > 0$, using the elementary inequality $(a+b)^p \le (1+\varepsilon)a^p + C_\varepsilon b^p$, we can bound the term on the right:
\[
\mathbb{E}^{\pi_0}[|P_t|^p \mid x] \;\le\; (1+\varepsilon)\rho_p |x|^p + C_\varepsilon K_0^p.
\]
Since $\rho_p < 1$, we can choose $\varepsilon$ small enough such that $\gamma := (1+\varepsilon)\rho_p < 1$. Letting $K = 1 + C_\varepsilon K_0^p$, we obtain the \emph{geometric drift condition} for the weight $w(x) = 1+|x|^p$:
\begin{equation*}
\label{eq:drift-inequality}
\mathbb{E}^{\pi_0}[w(P_t) \mid P_{t-1}=x] \;\le\; \gamma |x|^p + K \;\le\; \gamma w(x) + K.
\end{equation*}
This inequaility implies that the Markov chain induced by $\pi_0$ is geometrically ergodic in the $w$-norm. Since the optimal policy yields a value no worse than $\pi_0$, the optimal value function $V_\beta(x)$ is bounded by the value of $\pi_0$, which is finite and belongs to $\mathcal{B}_w$.

\medskip
\textit{1.4 Contraction of the Bellman Operator.}
We now verify that the Bellman operator $T_\beta$ is a contraction on $\mathcal{B}_w$. The operator is defined as:
\[
(T_\beta v)(x) \;:=\; \mathbb{E}_{A,B,J} \left[ \inf_{y \in I(x; A,B,J)} \big\{ c(x,y) + \beta v(y) \big\} \right],
\]
where $I(x; A,B,J) = [Ax+B+U, Ax+B+V]$ denotes the realized feasible interval.
Standard theorems for weighted-norm MDPs (e.g., \citealt[Theorem 8.3.6]{HernandezLermaLasserre1999}) state that $T_\beta$ is a contraction on $\mathcal{B}_w$ if two conditions are met:
\begin{enumerate}
    \item \textbf{Consistency of Cost:} The one-step cost is locally bounded by the weight function. Specifically, \cref{assum:stability}  guarantees $c(x,y) \le K(1+|x|^p+|y|^p)$. Since the interval boundaries grow linearly in $x$, any feasible $y$ satisfies $|y| \le |A_{max}x| + C$, so the cost is bounded by $C w(x)$.
    \item \textbf{Drift Condition:} The drift inequality \eqref{eq:drift-inequality} derived in Section 1.3 holds. This ensures that the ``energy'' of the system contracts in expectation, preventing the value function from exploding.
\end{enumerate}
Under these conditions, $T_\beta$ is a contraction mapping with respect to the weighted norm $\|\cdot\|_w$. By the Banach Fixed Point Theorem, a unique solution $V_\beta \in \mathcal{B}_w$ exists. This guarantees that the value function satisfies the polynomial growth bound $|V_\beta(x)| \le M(1+|x|^p)$.

\medskip
\noindent\textbf{Step 2: Convexity of the Value Function.}
We establish the convexity of the fixed point $V_\beta$ via a rigorous inductive argument on the Bellman operator $T_\beta$.

\smallskip
\textit{2.1 Base Case and Inductive Hypothesis.}
We proceed by induction. Let $\mathcal{C}_{cx}$ denote the set of convex functions in $\mathcal{B}_w$.
Initialize the value iteration with $V_0(x) \equiv 0$, which is trivially convex.
Assume as the inductive hypothesis that the function $V_n$ is convex (i.e., $V_n \in \mathcal{C}_{cx}$). We must show that the updated function $V_{n+1} = T_\beta V_n$ is also convex.

\smallskip
\textit{2.2 Joint Convexity of the State-Action Cost.}
The Bellman update is given by:
\[
V_{n+1}(x) \;=\; \mathbb{E}_{A,B,J} \left[ \inf_{y \in I(x; A,B,J)} \left\{ c(x,y) + \beta V_n(y) \right\} \right].
\]
Fix an arbitrary realization of the random parameters $\omega = (A, B, U, V)$. Define the cost function for this specific realization as:
\[
\Phi_\omega(x, y) \;:=\; c(x,y) + \beta V_n(y).
\]
We claim that $\Phi_\omega$ is \emph{jointly convex} in the pair $(x,y)$. To see this, consider the two terms separately:
\begin{enumerate}
    \item The immediate cost $c(x,y)$ is \emph{strictly jointly convex} by assumption.
    \item The future value term $\beta V_n(y)$ depends only on $y$. Since $V_n$ is convex by hypothesis, the map $(x,y) \mapsto V_n(y)$ is convex on $\mathbb{R}^2$ (it is constant along the $x$-axis and convex along the $y$-axis).
\end{enumerate}
Since the sum of convex functions is convex, $\Phi_\omega(x,y)$ is jointly convex in $(x,y)$.

\smallskip
\textit{2.3 Convexity of the Constraint Graph.}
The feasible interval for the realization $\omega$ is:
\[
\mathcal{G}_\omega \;:=\; \big\{ (x,y) \in \mathbb{R}^2 \mid Ax + B + U \le y \le Ax + B + V \big\}.
\]
Geometrically, this set represents the region between two parallel lines in the $(x,y)$-plane with slope $A$. Formally, it is the intersection of two closed half-spaces:
\[
\mathcal{G}_\omega \;=\; \{ (x,y) : y - Ax \ge B+U \} \cap \{ (x,y) : y - Ax \le B+V \}.
\]
Since half-spaces are convex sets, and the intersection of convex sets is convex, the feasible interval $\mathcal{G}_\omega$ is a convex subset of $\mathbb{R}^2$.

\smallskip
\textit{2.4 Preservation of Convexity under Minimization.}
Let $g_\omega(x)$ be the value of the minimization problem for this realization:
\[
g_\omega(x) \;:=\; \inf_{y} \big\{ \Phi_\omega(x,y) \mid (x,y) \in \mathcal{G}_\omega \big\}.
\]
This describes the partial minimization of a jointly convex function over a convex set. A fundamental result in convex analysis (see, e.g., \cite[Section 3.2.5]{boyd2004convex}) states that this operation preserves convexity.
Since the cost is non-negative, the infimum is finite, and thus $g_\omega(x)$ is a convex function of $x$.

\smallskip
\textit{2.5 Preservation under Expectation.}
The updated value function is the expectation over all possible realizations:
\[
V_{n+1}(x) \;=\; \mathbb{E}[g_\omega(x)] \;=\; \int g_\omega(x) \, d\mathbb{P}(\omega).
\]
Since the expectation operator is linear with a positive measure, and a non-negative weighted average (or integral) of convex functions is convex, $V_{n+1}(x)$ is convex.
This completes the inductive step: if $V_n$ is convex, then $V_{n+1}$ is convex.

\medskip
\textit{2.6 Convergence to the Fixed Point.}
In Step 1, we established that the Bellman operator is a contraction on the Banach space $\mathcal{B}_w$. Therefore, the sequence $V_n$ converges to the unique fixed point $V_\beta$ in the norm $\|\cdot\|_w$, which implies pointwise convergence.
Since the pointwise limit of a sequence of convex functions is convex, the unique discounted value function $V_\beta$ is convex.

\newpage

\medskip
\noindent\textbf{Step 3: The Average Cost Limit and Verification.}
Having established the properties of the discounted value function $V_\beta$, we pass to the limit to solve the long-run average cost problem.

\medskip
\textit{3.1 The Vanishing Discount Limit.}
We utilize the vanishing discount method. Define the \emph{differential value function} $h_\beta(x) := V_\beta(x) - V_\beta(0)$ and the average cost approximation $\lambda_\beta := (1-\beta)V_\beta(0)$.
Crucially, the \emph{geometric drift condition} derived in Step 1.3 ensures the stability required for unbounded costs (see \citet{HernandezLermaLasserre1999}). Specifically:
\begin{itemize}
    \item The sequence of normalized values $\{h_\beta\}_{\beta \in (0,1)}$ is uniformly bounded in the weighted norm on compact subsets of $\mathbb{R}$.
    \item The sequence of average costs $\{\lambda_\beta\}$ is bounded.
\end{itemize}
By the extended Arzela-Ascoli theorem, there exists a sequence $\beta_n \uparrow 1$ such that $h_{\beta_n}$ converges pointwise to a function $h(x)$ and $\lambda_{\beta_n} \to \lambda^*$. The limit function $h$ inherits convexity and the polynomial growth bound ($h \in \mathcal{B}_w$) from $V_\beta$.
Taking the limit in the discounted Bellman equation yields the \emph{Average Cost Optimality Equation} (ACOE):
\begin{equation}
\label{eq:ACOE-final}
\lambda^* + h(x)
\;=\;
\min_{y \in I(x)} \big\{ c(x,y) + \mathbb{E}[h(P_t)] \big\}.
\end{equation}

\medskip
\textit{3.2 Sufficiency of Stationary Policies.}
Standard results for weighted-norm MDPs (e.g., \cite{HernandezLermaLasserre1999}) establish that if a solution $(\lambda^*, h)$ to the ACOE exists and satisfies the growth condition $h \in \mathcal{B}_w$, then the stationary deterministic policy $\pi^*$ that selects the minimizer in the ACOE is average-cost optimal over the entire class of admissible policies.
Thus, finding the unique stationary solution to the ACOE is sufficient to characterize the globally optimal behavior.

\medskip
\textit{3.3 Verification (The Telescoping Sum).}
We now prove that solving the ACOE is sufficient for optimality. 
Let $\pi$ be any admissible policy. Under $\pi$, the state evolves as $P_t$. Applying the expectation operator $\mathbb{E}^\pi$ to the ACOE inequality at state $P_{t-1}$:
\[
\lambda^* + \mathbb{E}^\pi[h(P_{t-1})] \;\le\; \mathbb{E}^\pi \big[ c(P_{t-1}, P_t) + h(P_t) \big].
\]
The inequality becomes an equality if $\pi$ minimizes the RHS of the ACOE.
Summing this inequality from $t=1$ to $T$:
\[
T \lambda^* + \sum_{t=1}^T \mathbb{E}^\pi[h(P_{t-1})] \;\le\; \sum_{t=1}^T \mathbb{E}^\pi[c(P_{t-1}, P_t)] + \sum_{t=1}^T \mathbb{E}^\pi[h(P_t)].
\]
Observe that the sum of expectations telescopes. The term $\mathbb{E}^\pi[h(P_t)]$ appears on both sides (once as a cost-to-go, once as a previous value), cancelling out for all $t$ except the first ($t=0$) and the last ($t=T$). We simplify to:
\[
T \lambda^* + h(P_0) \;\le\; \mathbb{E}^\pi \left[ \sum_{t=1}^T c(P_{t-1}, P_t) \right] + \mathbb{E}^\pi[h(P_T)].
\]
Dividing by $T$ yields the fundamental performance bound:
\begin{equation}
\label{eq:performance-bound}
\lambda^* + \frac{h(P_0)}{T} \;\le\; \frac{1}{T} \mathbb{E}^\pi \left[ \sum_{t=1}^T c(P_{t-1}, P_t) \right] + \frac{\mathbb{E}^\pi[h(P_T)]}{T}.
\end{equation}
We now take the limit as $T \to \infty$:
\begin{enumerate}
    \item The term $h(P_0)/T$ vanishes to 0.
    \item The term $\frac{1}{T} \mathbb{E}^\pi [\sum c(P_{t-1}, P_t)]$ converges to the definition of the long-run average cost, $\mathcal{J}(\pi)$.
    \item The final term $\frac{\mathbb{E}^\pi[h(P_T)]}{T}$ must vanish. Since $h \in \mathcal{B}_w$, we have $|h(x)| \le C(1+|x|^p)$. For any policy with finite average cost, the $p$-th moment $\mathbb{E}[|P_T|^p]$ grows sub-linearly (or is bounded), ensuring $\lim_{T \to \infty} \frac{\mathbb{E}[h(P_T)]}{T} = 0$. (For policies with infinite cost, the inequality $\lambda^* \le \infty$ holds trivially).
\end{enumerate}
Thus, we obtain the lower bound $\lambda^* \le \mathcal{J}(\pi)$ for \emph{any} admissible policy.
Conversely, let $\pi^*$ be the policy that selects the minimizer in the ACOE. For this policy, the initial inequality is an equality. Following the same telescoping logic, we obtain $\lambda^* = \mathcal{J}(\pi^*)$.
Therefore, $\pi^*$ is optimal, and finding the minimizing policy in the ACOE is sufficient to prove the theorem.

\medskip
\noindent\textbf{Step 4: Characterization of the Optimal Policy.}
Finally, we determine the explicit functional form of the optimal policy by solving the minimization problem on the right-hand side of the ACOE.

\medskip
\textit{4.1 The Inner Minimization Problem.}
The optimal policy $\pi^*$ must select an action $y \in I_t$ that minimizes the sum of the immediate cost and the relative value of the next state.
Fix the current state $P_{t-1} = x$ and the realized interval $I_t = [L_t, R_t]$. The agent solves:
\[
\min_{y} \big\{ \psi_x(y) \big\} \quad \text{subject to} \quad L_t \le y \le R_t,
\]
where the objective function is defined as $\psi_x(y) := c(x,y) + h(y)$.

\medskip
\textit{4.2 Properties of the Objective Function.}
We analyze the function $\psi_x$:
\begin{enumerate}
    \item \textbf{Strict Convexity:} By \cref{assum:transitiondynamics}, the mapping $y \mapsto c(x,y)$ is strictly convex. In Step 3, we proved that $h(y)$ is convex. Since the sum of a strictly convex function and a convex function is strictly convex, $\psi_x$ is strictly convex on $\mathbb{R}$.
    \item \textbf{Coercivity (Growth):} We must ensure the minimum exists. By \cref{assum:cost} (Monotonicity), the cost $c(x,y)$ grows unbounded as $|y| \to \infty$. Furthermore, since the original costs are non-negative, the value function is bounded below by zero ($V_\beta \ge 0$), implying $h(y) = V_\beta(y) - V_\beta(0) \ge -V_\beta(0)$. Thus, the sum $\psi_x(y) \to \infty$ as $|y| \to \infty$.
\end{enumerate}
A strictly convex, coercive function on $\mathbb{R}$ possesses a unique global minimizer. We define this unconstrained minimizer as the ``Reference Target'' $r(x)$:
\[
r(x) \;:=\; \arg\min_{z \in \mathbb{R}} \big\{ c(x,z) + h(z) \big\}.
\]
Crucially, $r(x)$ depends \emph{only} on the current action $x$ and the value function $h$; it is independent of the interval realization $I_t$.

\medskip
\textit{4.3 The Projection Argument.}
We now re-introduce the constraint $y \in [L_t, R_t]$.
Because $\psi_x$ is strictly convex with a unique global minimum at $r(x)$, the function satisfies a strict monotonicity property relative to the minimum: it is strictly decreasing for all $y < r(x)$ and strictly increasing for all $y > r(x)$.
The constrained minimizer $y^*$ is therefore determined solely by the position of the interval relative to $r(x)$:
\begin{itemize}
    \item \textbf{Case 1: $r(x) \in [L_t, R_t]$.} The global minimum is feasible. Since it minimizes $\psi_x$ over all $\mathbb{R}$, it minimizes $\psi_x$ over the subset $[L_t, R_t]$. Thus $y^* = r(x)$.
    \item \textbf{Case 2: $r(x) < L_t$.} The interval lies entirely to the right of the minimum. On the domain $[L_t, \infty)$, the function $\psi_x$ is strictly increasing. Therefore, the minimum occurs at the leftmost boundary: $y^* = L_t$.
    \item \textbf{Case 3: $r(x) > R_t$.} The interval lies entirely to the left of the minimum. On the domain $(-\infty, R_t]$, the function $\psi_x$ is strictly decreasing. Therefore, the minimum occurs at the rightmost boundary: $y^* = R_t$.
\end{itemize}
This case-by-case logic corresponds exactly to the definition of the projection operator $\Pi_{[L,R]}(z) = \max(L, \min(R, z))$.
Thus, for \emph{any} realization of the interval, the optimal action is:
\[
P_t^* \;=\; \Pi_{I_t}\big(r(P_{t-1})\big).
\]

\smallskip
\textit{4.4 Optimality of a Reference Policy.}
We have constructed a policy $P_t = \Pi_{I_t}(r(P_{t-1}))$ that minimizes the RHS of the Average Cost Optimality Equation state-by-state. By the verification argument in Step 3, this policy is globally optimal for minimizing the long-run adjustment cost.
Since this policy belongs to the class $\mathcal{M}^{\mathrm{ref}}$, we have proven both the existence of an optimal policy and the structural sufficiency of reference-price policies.
Furthermore, since the objective $\psi_x(y)$ is strictly convex, the minimizer is unique. Thus, the reference policy derived here is the unique optimal policy

\smallskip
\noindent\textbf{Step 5: Structural Properties of the Reference Map.}
We now derive the structural properties (continuity, monotonicity, contraction, fixed point) of the optimal reference target $r(x)$. This requires refining our analysis of the relative value function $h$ to establish strict convexity and utilizing subgradient optimality conditions.

\smallskip
\textit{5.1 Strict Convexity of the Value Function.}
In Step 2, we established that the value function $h$ is convex. We now strengthen this to \emph{strict} convexity.
Recall that $h(x) = \mathbb{E}[\psi_{I_t}(x)] - \lambda^*$. Due to the stochastic support of the noise terms, for any state $x$, there is a strictly positive probability that the unconstrained minimizer $r(x)$ falls outside the realized interval $I_t$. In such events, the optimal policy is forced to the boundary.
Consider a realization where the upper constraint is active ($P_t = R_t$). The realized cost function becomes:
\[
\psi_{I_t}(x) \;=\; c(x, R_t) + h(R_t).
\]
Substituting the linear transition dynamics $R_t = A_t x + B_t + V_t$, the immediate cost term is $f(x) = c(x, A_t x + K_t)$.
We analyze the curvature of this term:
\begin{itemize}
    \item If $c$ is \emph{strictly jointly convex} (\cref{assum:cost} (Strict Convexity)), the composite function $f(x)$ is strictly convex in $x$ for any $A_t$, because the trajectory $(x, A_t x + K_t)$ traces a line in the domain. Strict convexity is preserved along lines unless the function is linear along that specific direction, which is ruled out by the growth conditions of $c$.
    \item In the specific case of \emph{translation invariant} costs (e.g., $c(x,y)=\phi(y-x)$), the term simplifies to $\phi((A_t-1)x + K_t)$. This is strictly convex provided $A_t \neq 1$. Since \cref{assum:stability} requires $\mathbb{E}[|A_t|^p] < 1$, the variable $A_t$ cannot be identically 1.
\end{itemize}
Since $h(x)$ is the expectation of these valuations, and the event of hitting a boundary occurs with positive probability, $h$ inherits this strictly positive curvature. Therefore, $h$ is \emph{strictly} convex on $\mathbb{R}$.

\smallskip
\textit{5.2 Optimality Conditions via Subgradients.}
To characterize the optimal target $r(x)$ without assuming differentiability of the value function $h$, we employ subgradient calculus.
Recall that $r(x)$ minimizes the strictly convex function $\phi_x(z) = c(x,z) + h(z)$. The necessary and sufficient condition for optimality is that zero belongs to the subdifferential of the objective function:
\begin{equation}
\label{eq:FOC-subgradient}
0 \;\in\; \partial_y c(x, r(x)) + \partial h(r(x)),
\end{equation}
where $\partial_y c$ denotes the partial derivative of the cost with respect to its second argument (which exists since $c \in C^2$) and $\partial h$ denotes the subdifferential of $h$.
This inclusion implies that there exists a specific subgradient $\xi \in \partial h(r(x))$ such that:
\[
\xi \;=\; - \partial_y c(x, r(x)).
\]

\smallskip
\textit{5.3 Continuity of the Reference Map.}
We explicitly prove that $r(x)$ is a continuous function.
Consider the minimization problem defining the target:
\[
r(x) \;=\; \arg\min_{z \in \mathbb{R}} \big\{ c(x,z) + h(z) \big\}.
\]
The objective function $\Phi(x,z) = c(x,z) + h(z)$ is jointly continuous in $(x,z)$ (since $c$ is $C^2$ and $h$ is convex/finite). Furthermore, as established in Step 4.2, the objective is coercive in $z$, effectively restricting the minimization to a compact set for any local neighborhood of $x$.
Since the objective is strictly convex in $z$ (Step 4.2), the minimizer is unique.
By Berge's Maximum Theorem applied to a strictly convex objective, the unique minimizer is a continuous function of the parameter $x$. Thus, $r \in C^0(\mathbb{R})$.
\end{proof}

\subsection{Proof of \Cref{theorem:non-expansive}}\label{ssec:proof-non-expansive}

\begin{proof}

This proof builds on the notation and results established in the proof of \Cref{theorem:non-expansive} (\Cref{ssec:ProofTheorem1}).

We first prove that the target function is \emph{strictly} increasing. 
By assumption, the cost function is submodular, that is, $c_{xy} < 0$.
Consider two distinct states $x_1 < x_2$. Let $r_1 = r(x_1)$ and $r_2 = r(x_2)$.
We revisit the structural balance equation derived from the optimality conditions and the Multivariate Mean Value Theorem:
\begin{equation}
\label{eq:structural_balance_strict}
c_{yy}(\tilde{x}, \tilde{r})(r_2 - r_1) + (\xi_2 - \xi_1) \;=\; -c_{yx}(\tilde{x}, \tilde{r})(x_2 - x_1).
\end{equation}
We proceed by contradiction. Suppose $r_2 \le r_1$ (Non-Increasing).
\begin{itemize}
    \item \textbf{RHS Analysis:} Since $x_2 > x_1$, the term $(x_2 - x_1)$ is strictly positive. By \Cref{assum:cost}, the cross-derivative $c_{yx}$ is strictly negative. Thus, the RHS is strictly positive.
    \item \textbf{LHS Analysis:}
    Since $r_2 \le r_1$, the term $(r_2 - r_1)$ is non-positive. Since $c_{yy} > 0$, the first term is non-positive.
    Since $h$ is strictly convex, the subgradient operator is strictly monotone. Thus $r_2 \le r_1 \implies \xi_2 \le \xi_1$, making the second term non-positive.
    Thus, the LHS is non-positive.
\end{itemize}
We have a contradiction: (Non-Positive) = (Strictly Positive). Therefore, the assumption $r_2 \le r_1$ is false. We must have $r_2 > r_1$.
This proves that the reference target $r(x)$ is strictly increasing.

Next, we prove that the policy is non-expansive, i.e., $|r(x_2) - r(x_1)| < |x_2 - x_1|$.
Using strict monotonicity, we know $r_2 > r_1$.
Consequently, by the strict convexity of $h$, the subgradient difference is strictly positive: $\xi_2 - \xi_1 > 0$.
We return to the balance equation \eqref{eq:structural_balance_strict}. 
If we drop the strictly positive term $(\xi_2 - \xi_1)$ from the LHS, the equality becomes a strict inequality:
\[
c_{yy}(\tilde{x}, \tilde{r})(r_2 - r_1) \;<\; -c_{yx}(\tilde{x}, \tilde{r})(x_2 - x_1).
\]
Dividing by the strictly positive curvature $c_{yy}$:
\[
r_2 - r_1 \;<\; \left( \frac{-c_{yx}(\tilde{x}, \tilde{r})}{c_{yy}(\tilde{x}, \tilde{r})} \right) (x_2 - x_1).
\]
We now invoke the condition $|c_{xy}| \le c_{yy}$. This implies that at the intermediate point, the ratio $-c_{yx} / c_{yy} \le 1$.
Substituting this bound:
\[
r_2 - r_1 \;<\; 1 \cdot (x_2 - x_1).
\]
Combining this with Step 5.4 ($r_2 > r_1$), we have $0 < r_2 - r_1 < x_2 - x_1$.
This confirms that the reference map is a strict contraction.

\medskip
A fixed point $p^\ast$ satisfies $0 \in \partial_y c(p^\ast, p^\ast) + \partial h(p^\ast)$.
By \cref{assum:cost} (No Cost for Inaction), $\partial_y c(x,x) = 0$. Thus, the condition simplifies to $0 \in \partial h(p^\ast)$.
This states that $p^\ast$ must be a global minimizer of the relative value function $h$.
\begin{itemize}
    \item \textbf{Existence:} As shown in Step 4.2 of the proof of \Cref{thm:optimal_policy_adjustment}, the coercivity of $c$ ensures $h$ is coercive. A continuous, coercive function on $\mathbb{R}$ attains a global minimum.
    \item \textbf{Uniqueness:} Since $h$ is strictly convex (Step 5.1), it has a unique global minimizer.
\end{itemize}
Thus, there exists a unique state $p^\ast$ that serves as the fixed point of the reference policy.
\end{proof}

\subsection{Proof of \Cref{theorem:rwinertia}}\label{ssec:ProofTheorem2}

\begin{proof}
The proof proceeds by formalizing the problem in terms of increments, showing that the dynamic optimization decouples into independent static problems, and utilizing the properties of convex minimization to identify the unique solution.

\medskip
\noindent\textbf{Step 1: Formalization in Increments.}
Let $(\mathcal{F}_t)_{t\ge 0}$ be the natural filtration generated by $P_0$ and the base intervals $(J_k)_{k \ge 1}$. 
In the random walk environment ($A_t \equiv 1$), the constraint $P_t \in P_{t-1} + J_t$ is equivalent to constraining the action increment $\Delta P_t := P_t - P_{t-1}$ to lie within $J_t$.
By assumption, the adjustment cost is translation invariant: $c(P_{t-1}, P_t) = \phi(\Delta P_t)$.
Thus, the global optimization problem is to minimize:
\[
\mathcal{J}(\pi) \;:=\; \limsup_{T\to\infty}\frac{1}{T} \sum_{t=1}^T \mathbb E [\phi(\Delta P_t)],
\]
subject to $\Delta P_t \in J_t$ almost surely.

\medskip
\noindent\textbf{Step 2: Decoupling of the Optimization Problem.}
In general MDPs, the current action affects the state, which constrains future actions. Here, however, the problem decouples entirely:
\begin{enumerate}
    \item \textbf{Feasibility Independence:} The feasibility of the future increment $\Delta P_{t+1} \in J_{t+1}$ depends only on the exogenous random interval $J_{t+1}$. It is statistically independent of the current action level $P_t$ or the current action $\Delta P_t$.
    \item \textbf{Cost Independence:} The future cost $\phi(\Delta P_{t+1})$ depends only on the future increment. Unlike the mean-reverting case, there is no "preferred" action level; shifting the entire action path by a constant does not alter the cost structure. Thus, the marginal value of the state is constant ($\nabla h \equiv 0$).
\end{enumerate}
Because the current action $\Delta P_t$ exerts no influence on future feasible intervals or future cost functions, the global minimization of the sum reduces to minimizing the expected cost $\mathbb{E}[\phi(\Delta P_t)]$ independently for each period $t$.

\medskip
\noindent\textbf{Step 3: The Static Lower Bound.}
We solve the static minimization problem pointwise. 
Fix a realization of the base interval $J_t = [U_t, V_t]$. The agent solves:
\[
\min_{\delta} \phi(\delta) \quad \text{subject to} \quad \delta \in J_t.
\]
By \Cref{assum:cost} and translation invariance, $\phi$ is strictly convex with a global unconstrained minimum at $\delta = 0$.
The minimizer of a strictly convex function over a closed interval is simply the point in the interval closest to the global minimum. 
Mathematically, the unique solution $\delta^*(J_t)$ is the projection of $0$ onto $J_t$:
\[
\delta^*(J_t) \;=\; \Pi_{J_t}(0) \;=\;
\begin{cases}
U_t, & \text{if } U_t > 0 \quad (\text{Interval is to the right}), \\
0, & \text{if } 0 \in [U_t, V_t] \quad (\text{Zero is feasible}), \\
V_t, & \text{if } V_t < 0 \quad (\text{Interval is to the left}).
\end{cases}
\]
This establishes a universal lower bound: for any admissible policy, $\phi(\Delta P_t) \ge \phi(\delta^*(J_t))$ almost surely.

\medskip
\noindent\textbf{Step 4: Optimality of the \RR.}
The \RR sets the action to $P_t = \Pi_{I_t}(P_{t-1})$.
Substituting $I_t = P_{t-1} + J_t$ and utilizing the translation invariance of the projection operator ($\Pi_{x+J}(x) = x + \Pi_{J}(0)$), the implied increment is:
\[
\Delta P_t^{\mathrm{Inertia}} \;=\; \Pi_{P_{t-1}+J_t}(P_{t-1}) - P_{t-1} \;=\; \Pi_{J_t}(0).
\]
The \RR generates the optimal increment $\delta^*(J_t)$ for every realization of the constraints. Therefore, it achieves the pointwise minimum cost almost surely and minimizes the long-run average cost.

\medskip
\noindent\textbf{Step 5: Uniqueness.}
Uniqueness follows from the strict convexity of $\phi$.
Suppose there exists another optimal policy $\tilde{\pi}$ producing increments $\Delta \tilde{P}_t$. For this policy to achieve the lower bound, it must satisfy $\mathbb{E}[\phi(\Delta \tilde{P}_t)] = \mathbb{E}[\phi(\delta^*(J_t))]$.
Since $\delta^*(J_t)$ is the unique minimizer of $\phi$ on $J_t$ (due to strict convexity), any deviation $\Delta \tilde{P}_t \neq \delta^*(J_t)$ on a set of positive measure implies strict inequality $\phi(\Delta \tilde{P}_t) > \phi(\delta^*(J_t))$.
Taking expectations, this would imply a strictly higher cost, contradicting optimality. Thus, the \RR is the unique optimal policy.
\end{proof}

\subsection{Proof of \Cref{theorem:varanchor}}\label{ssec:ProofTheorem3}

\begin{proof}
The proof is constructive and proceeds in five steps. 
We use the variational characterization of variance to decompose the global non-linear optimization problem into a family of tractable average-cost Markov Decision Processes (MDPs), each indexed by a scalar target parameter.

\smallskip
\noindent\textbf{Preliminaries and Stability.}
The state space is $\mathsf S = \mathbb{R}$. The feasible intervals evolve according to the affine dynamics $I_t(p) = [A_t p + B_t + U_t, \, A_t p + B_t + V_t]$.
We invoke the model's stability condition, $\mathbb{E}[A_t^p] < 1$. 
By \Cref{proposition:universal_stability}, this condition ensures that the second moment of the action process is uniformly bounded for \emph{any} admissible policy $\pi$:
\[
\sup_{t \ge 0} \mathbb{E}^\pi[P_t^2] < \infty.
\]
This boundedness is crucial: it guarantees that all Mean Squared Error (MSE) objectives defined below are finite, ensuring the optimization problems are well-posed in the space of square-integrable random variables.

\smallskip
\noindent\textbf{Step 0: Variational Formulation.}
Let $\pi$ be an admissible policy. 
We seek to minimize the stationary variance $\mathrm{Var}(\pi)$. 
The variance minimization problem is:
\[
\mathcal{V}^\ast \;=\; \inf_{\pi} \left( \inf_{c \in \mathbb{R}} J(\pi, c) \right).
\]
Since the minimization is over the product space of policies and scalars ($\Pi \times \mathbb{R}$), we can interchange the order of operations:
\begin{equation}
\label{eq:interchange}
\inf_{\pi} \inf_{c \in \mathbb{R}} J(\pi, c)
\;=\;
\inf_{c \in \mathbb{R}} \left( \inf_{\pi} J(\pi, c) \right).
\end{equation}
This decomposes the difficult global problem into two simpler sequential tasks:
\begin{enumerate}
    \item \textbf{Inner Control Problem:} Fix $c$ and find the policy $\pi_c$ that minimizes the MSE relative to $c$.
    \item \textbf{Outer Scalar Problem:} Minimize the resulting value function $F(c) := J(\pi_c, c)$ with respect to $c$.
\end{enumerate}

\smallskip
\noindent\textbf{Step 1: Solving the Fixed-Target MDP.}
Fix the target scalar $c \in \mathbb{R}$. 
The inner problem is to find a policy minimizing the average cost with the one-period cost function $g_c(y) = (y-c)^2$.
This is an infinite-horizon average-cost MDP with state $P_{t-1}=x$ and action $P_t=y \in I_t(x)$.

\smallskip
\textit{1.1 Weighted Norm Consistency.}
Since the cost function $(y-c)^2$ is unbounded, we work within the weighted Banach space $\mathcal{B}_w$ with weight $w(x) = 1+x^2$.
We must verify that the cost function does not grow faster than the weight. For any feasible action $y$, we have:
\[
g_c(y) = (y-c)^2 \le 2y^2 + 2c^2 \le 2(1+y^2) + 2c^2 \le C(c) \cdot w(y).
\]
Since the transition kernel satisfies the geometric drift condition for $w$ (as established in \Cref{proposition:universal_stability}), and the cost is dominated by $w$, the standard theory for unbounded-cost MDPs (e.g. \cite{HernandezLermaLasserre1999}) applies.

\smallskip
\textit{1.2 Optimality Equation.}
There exists a scalar optimal cost $\lambda_c^*$ and a relative value function $h_c \in \mathcal{B}_w$ satisfying the Average Cost Optimality Equation (ACOE):
\begin{equation}
\label{eq:ACOE-var}
\lambda_c^* + h_c(x)
\;=\;
\mathbb E \left[ \min_{y \in I_t(x)} \big\{ (y-c)^2 + h_c(y) \big\} \right],
\quad \forall x \in \mathbb{R}.
\end{equation}
Furthermore, the stationary deterministic policy $M_c$ that selects the minimizer in the RHS of \eqref{eq:ACOE-var} is optimal for the fixed target $c$.

\smallskip
\noindent\textbf{Step 2: Convexity of the Bias Function.}
To characterize the policy $M_c$, we determine the shape of $h_c$. We prove it is convex by analyzing the finite-horizon/discounted approximations.

Fix a discount factor $\beta \in (0,1)$ and consider the discounted value iteration $V_{n+1} = T V_n$.
Suppose $V_n$ is convex. The update for a specific realization of parameters $\omega=(A, B, U, V)$ is:
\[
g_\omega(x) \;=\; \min_{y} \big\{ (y-c)^2 + \beta V_n(y) \big\} \quad \text{s.t.} \quad y \in [Ax+B+U, Ax+B+V].
\]
\begin{enumerate}
    \item \textbf{Joint Convexity:} The objective $(y-c)^2 + \beta V_n(y)$ is jointly convex in $(x,y)$ (it is strictly convex in $y$ and constant in $x$).
    \item \textbf{Convex Constraint:} The feasible interval $\{(x,y) : Ax+B+U \le y \le Ax+B+V\}$ is the intersection of closed half-spaces, hence convex.
    \item \textbf{Preservation:} Minimizing a jointly convex function over a convex set results in a convex function of the parameter $x$. Thus, $g_\omega(x)$ is convex.
\end{enumerate}
The full update $V_{n+1}(x) = \mathbb{E}[g_\omega(x)]$ is an expectation of convex functions, which is convex. By induction, the discounted value function $V_{\beta,c}$ is convex.
The average-cost bias function is obtained as the limit $h_c(x) = \lim_{\beta \uparrow 1} (V_{\beta,c}(x) - V_{\beta,c}(0))$. Since convexity is preserved under pointwise limits, $h_c$ is convex.

\medskip
\noindent\textbf{Step 3: Characterization of the Policy $M_c$.}
We explicitly solve the minimization in the ACOE:
\[
\min_{y \in [L_t, R_t]} \big\{ \psi_c(y) \big\}, \quad \text{where } \psi_c(y) = (y-c)^2 + h_c(y).
\]
Crucially, the objective function $\psi_c(y)$ depends on the target $c$ and the continuation value $h_c$, but is independent of the current state $x$. 
The state $x$ enters the problem only through the constraint bounds $[L_t, R_t]$.

The function $\psi_c(y)$ is strictly convex (as the sum of the strictly convex immediate cost $(y-c)^2$ and the convex continuation value $h_c$). 
By coercivity, it has a unique global minimizer on $\mathbb{R}$, which we denote by $z_c^\ast$:
\[
z_c^\ast \;:=\; \arg\min_{z \in \mathbb{R}} \big\{ (z-c)^2 + h_c(z) \big\}.
\]
This defines an Anchor Policy: for a fixed variance center $c$, the optimal policy $M_c$ steers the action toward a constant target $z_c^\ast$.

\medskip
\noindent\textbf{Step 4: Existence of the Optimal Target.}
We return to the outer problem: minimizing $F(c) := \inf_{\pi} J(\pi, c) = J(M_c, c)$. We must show that an optimal scalar $c^\ast$ exists.
We argue that if we choose a target $c$ extremely far from the natural center of the market, the cost must explode.
Recall the dynamics $P_t = A_t P_{t-1} + B_t + \xi_t$, where $\xi_t$ is the policy's adjustment.
The stationary mean $\mu_M$ of \emph{any} stationary policy $M$, and hence of all Anchoring Policies, satisfies the balance equation $\mu_M = \mathbb{E}[A]\mu_M + \mathbb{E}[B] + \mathbb{E}[\xi]$.
Since the adjustment $\xi_t$ is constrained to lie within the noise interval $[U_t, V_t]$, its expectation is uniformly bounded by the model primitives: $\mathbb{E}[U] \le \mathbb{E}[\xi] \le \mathbb{E}[V]$.
Solving for the mean:
\[
\mu_M \;=\; \frac{\mathbb{E}[B] + \mathbb{E}[\xi]}{1 - \mathbb{E}[A]}.
\]
This implies that the set of achievable stationary means is a bounded interval $[-R, R]$.
Now consider the cost $F(c)$. Using the MSE identity:
\[
F(c) \;=\; J(M_c, c) \;=\; \mathrm{Var}(M_c) + (\mu_{M_c} - c)^2 \;\ge\; (\mu_{M_c} - c)^2.
\]
As $|c| \to \infty$, the distance between $c$ and the bounded mean $\mu_{M_c}$ grows linearly, so the cost grows quadratically. Thus $F(c)$ is coercive.
By the Generalized Weierstrass Theorem, a lower semicontinuous coercive function attains its global minimum. Let $c^\ast$ be a minimizer.

\medskip
\noindent\textbf{Step 5: Global Optimality.}
We construct the global policy and verify it minimizes variance.
Let $m^\ast := z_{c^\ast}^\ast$ be the optimal dynamic target associated with the optimal parameter $c^\ast$.
Define the policy $M^\ast$ as the Anchoring Rule targeting $m^\ast$: $P_t \;=\; \Pi_{I_t}(m^\ast).$
We chain the optimality conditions derived in Steps 0-4:
\begin{align*}
\mathrm{Var}(M^\ast)
&\;\le\; J(M^\ast, c^\ast) & (\text{Property of Variance: } \sigma^2 \le \text{MSE}_c) \\
&\;=\; \inf_{\pi} J(\pi, c^\ast) & (\text{Step 3: } M^\ast \text{ is optimal for } c^\ast) \\
&\;=\; F(c^\ast) & (\text{Definition of } F) \\
&\;=\; \inf_{c} F(c) & (\text{Step 4: } c^\ast \text{ minimizes } F) \\
&\;=\; \inf_{c} \inf_{\pi} J(\pi, c) & (\text{Expansion}) \\
&\;=\; \inf_{\pi} \inf_{c} J(\pi, c) & (\text{Step 0: Interchange}) \\
&\;=\; \inf_{\pi} \mathrm{Var}(\pi). & (\text{Variational Identity})
\end{align*}
This equality confirms that the Anchoring Rule $M^\ast$ achieves the global minimum variance among all admissible policies.
\end{proof}

\subsection{Proof of \Cref{theorem:optanchor}}\label{ssec:ProofTheorem4}

\begin{proof}
By \Cref{theorem:varanchor}, the variance-minimizing policy is an anchoring rule $P_t = \Pi_{I_t}(m)$ for some scalar $m \in \mathbb{R}$. 
Furthermore, the proof of \Cref{theorem:varanchor} established that the minimal achievable variance is exactly the global minimum of the function:
\[
J(m) \;:=\; \mathbb E \big[ (\Pi_{I_t}(m) - m)^2 \big].
\]
Our task is to characterize the unique minimizer $m^\ast$ of this strictly convex function.

\medskip
\noindent\textbf{Step 1: Differentiability and Convexity.}
Let $I = [L, R]$ be a generic realization of the random constraint interval. Define the realized cost function $j_I(m) := (\Pi_I(m) - m)^2$, which represents the squared Euclidean distance from $m$ to the interval $I$.
Explicitly:
\[
j_I(m) \;=\;
\begin{cases}
(L-m)^2 & \text{if } m < L, \\
0 & \text{if } L \le m \le R, \\
(R-m)^2 & \text{if } m > R.
\end{cases}
\]
This function is continuously differentiable with respect to $m$. Its derivative is:
\[
j_I'(m) \;=\; 2(m - \Pi_I(m)).
\]
If $m < L$, derivative is $2(m-L)$; if $L \le m \le R$, derivative is 0; if $m > R$, derivative is $2(m-R)$).

We now differentiate the expected cost $J(m) = \mathbb{E}[j_I(m)]$.
Since the distance $|\Pi_I(m) - m|$ is bounded by $|m| + |L| + |R|$, and the constraint boundaries have finite first moments, the derivative is dominated by an integrable function. By the Dominated Convergence Theorem, we can differentiate under the expectation:
\begin{equation}
\label{eq:J-prime-direct}
J'(m) \;=\; \mathbb E \big[ j_I'(m) \big] \;=\; 2 \mathbb E \big[ m - \Pi_{I_t}(m) \big].
\end{equation}
To establish convexity, we examine the second derivative. The projection function $\Pi_I(m)$ has a slope of 1 strictly inside the interval and 0 outside. Thus, its derivative is almost everywhere $\Pi_I'(m) = \mathbf{1}_{\{L < m < R\}}$.
Substituting this into the expression for the second derivative of the realized cost $j_I''(m) = 2(1 - \Pi_I'(m))$, we find:
\[
j_I''(m) \;=\; 2\big(1 - \mathbf{1}_{\{L < m < R\}}\big) \;=\; 2 \cdot \mathbf{1}_{\{m \notin (L, R)\}}.
\]
The term $\mathbf{1}_{\{m \notin (L, R)\}}$ is the indicator that the anchor $m$ lies outside the feasible interval, meaning the constraint is binding.
Taking expectations, the second derivative of the aggregate cost is directly proportional to the probability that the anchor falls outside the random interval:
\[
J''(m) \;=\; 2 \cdot \mathbb{E}\big[ \mathbf{1}_{\{m \notin (L, R)\}} \big] \;=\; 2 \cdot \mathbb{P}\big( \{m \le L\} \cup \{m \ge R\} \big).
\]
Assuming the interval distribution is non-degenerate, the probability of the constraint being active is strictly positive, implying $J''(m) > 0$. Thus, $J(m)$ is strictly convex and coercive, guaranteeing a unique global minimizer.

\medskip
\noindent\textbf{Step 2: The Self-Consistency Condition.}
The optimal anchor $m^\ast$ is defined by the first-order condition $J'(m^\ast) = 0$.
Substituting \eqref{eq:J-prime-direct}:
\[
2 \mathbb E \big[ m^\ast - \Pi_{I_t}(m^\ast) \big] \;=\; 0 \quad \implies \quad m^\ast \;=\; \mathbb E \big[ \Pi_{I_t}(m^\ast) \big].
\]
This proves the second claim of the theorem: the target must equal the expected stationary action that results from using that target.

\medskip
\noindent\textbf{Step 3: The Boundary-Balance Equation.}
We expand the term $\Pi_{I_t}(m^\ast)$ using indicator functions for the three regions relative to $m^\ast$:
\[
\Pi_{I_t}(m^\ast) \;=\; L_t \mathbf{1}_{\{L_t > m^\ast\}} + m^\ast \mathbf{1}_{\{L_t \le m^\ast \le R_t\}} + R_t \mathbf{1}_{\{R_t < m^\ast\}}.
\]
Substituting this into the condition $m^\ast = \mathbb{E}[\Pi_{I_t}(m^\ast)]$ yields:
\[
m^\ast \;=\; \mathbb{E}\big[ L_t \mathbf{1}_{\{L_t > m^\ast\}} + m^\ast \mathbf{1}_{\{L_t \le m^\ast \le R_t\}} + R_t \mathbf{1}_{\{R_t < m^\ast\}} \big].
\]
Subtracting $m^\ast$ (written as $m^\ast(\mathbf{1}_{\{L > m\}} + \mathbf{1}_{\{L \le m \le R\}} + \mathbf{1}_{\{R < m\}})$) from both sides cancels the middle term and centers the boundary terms:
\[
0 \;=\; \mathbb{E}\Big[ (L_t - m^\ast)\mathbf{1}_{\{L_t > m^\ast\}} \Big] \;+\; \mathbb{E}\Big[ (R_t - m^\ast)\mathbf{1}_{\{R_t < m^\ast\}} \Big].
\]
This is exactly Eq.~\eqref{eq:anchor-balance}.

\medskip
\noindent\textbf{Step 4: Symmetry Case.}
Suppose the joint distribution of $(L_t, R_t)$ is symmetric about $c$. Let $\delta = m - c$. Due to symmetry, the cost function $J(m)$ satisfies $J(c+\delta) = J(c-\delta)$.
Specifically, if we shift the anchor away from $c$, the expected squared distance to a symmetric random interval increases identically in both directions.
The unique minimizer of a strictly convex symmetric function is its center of symmetry. Thus, $m^\ast = c$.
\end{proof}

\subsection{Proof of Proposition 1}\label{ssec:ProofBestResponse}

\begin{proof}
Let $\Delta(v) := U_{stat}(v) - \max_{b < R} U_{mono}(b; v)$ denote the utility difference between the optimal Inertia strategy and the optimal Monopsony strategy. We establish the result by showing that $\Delta(v)$ satisfies a single-crossing property.

First, we apply the Envelope Theorem to determine the slope of the value functions. For the Monopsony profit $V_{mono}(v) = \max_b (v-b)G(b)$, the derivative is simply the probability of trade at the optimum, $V'_{mono}(v) = G(b^m(v))$. For the Inertia profit $U_{stat}(v)$, differentiation yields $U'_{stat}(v) = G(R) + \int_R^v g(c)dc = G(v)$. The rate of change of the utility difference is therefore $\Delta'(v) = G(v) - G(b^m(v))$. In the monopsony regime, the Buyer strictly shades their bid to extract surplus, implying $b^m(v) < v$. Since $G(\cdot)$ is strictly increasing, it follows that $G(v) > G(b^m(v))$, and thus $\Delta'(v) > 0$ for all $v$.

Second, we check the boundary conditions. At $v=R$, the Inertia strategy yields zero surplus, $U_{stat}(R) = 0$. The Monopsony strategy, however, yields strictly positive profit because the buyer can set an optimal bid $b^m(R) < R$ and extract surplus from sellers with $c < b^m(R)$. Thus, $\Delta(R) < 0$. As $v \to 1$, the Monopsony strategy entails a permanent efficiency loss due to bid shading, whereas the Inertia strategy approaches first-best efficiency (trade occurs whenever $c \le v$). Consequently, $\Delta(v) > 0$ for sufficiently high $v$. Since $\Delta(v)$ is continuous, strictly monotonic, starts negative, and ends positive, there exists a unique $\hat{v} \in (R, 1]$ such that $\Delta(\hat{v}) = 0$.
\end{proof}

\subsection{Proof of \Cref{proposition:BNE}}\label{ssec:ProofBNE}

\begin{proof}

\textbf{Impossibility of a non-degenerate BNE.} We proceed by contradiction. Assume there exists a Bayesian Nash Equilibrium where the strategies $b: [0,1] \to \mathbb{R}$ and $s: [0,1] \to \mathbb{R}$ are continuous and strictly increasing, and where the event $\{P = R\}$ occurs with strictly positive probability.

\medskip
\noindent\textbf{Step 1: The Domain of Strategies.}
Let $H_S$ denote the distribution of asks induced by the seller's strategy. Since $s(c)$ is continuous and strictly increasing, $H_S$ admits a continuous density $h_S$.
The condition that trade occurs at $R$ with positive probability implies that the event $\{s(C) \le R \le b(V)\}$ has positive measure. Given strictly increasing strategies (which preclude atoms at $R$), this requires the supports of the strategies to strictly overlap $R$.
Specifically, the range of equilibrium bids must include an interval $(R, R+\delta]$ and the range of equilibrium asks must include an interval $[R-\delta, R)$. Consequently, $H_S(R) > 0$ and $h_S(R) > 0$.

\medskip
\noindent\textbf{Step 2: Optimality Above the Reference Price.}
Consider a buyer who submits a bid $b$ in the range $(R, R+\delta)$.
The trade volume splits into two regions:
\begin{itemize}
    \item \textbf{Infra-marginal sellers ($s \le R$):} The interval is $[s, b]$. Since $s \le R < b$, the reference price is strictly interior. The projection rule sets the price to $P = \Pi_{[s, b]}(R) = R$.
    \item \textbf{Marginal sellers ($R < s \le b$):} The interval is $[s, b]$. The reference price is below the interval. The projection rule sets the price to the lower bound $P = s$.
\end{itemize}
The expected utility for a buyer with valuation $v$ is:
\[ U_{right}(b; v) = \int_{\underline{s}}^{R} (v - R) \, dH_S(s) + \int_{R}^{b} (v - s) \, dH_S(s). \]
Differentiating with respect to $b$ (using Leibniz's Rule), we observe that the first term is independent of $b$. The marginal utility depends only on the second term:
\[ \frac{\partial U}{\partial b} = (v - b)h_S(b). \]
For the strategy to be optimal in this region, the First Order Condition requires $v = b$.
Thus, any buyer bidding above $R$ must bid truthfully: $b(v) = v$.
Taking the limit as the bid approaches $R$ from above, continuity implies that the valuation of the buyer bidding exactly $R$ must be $R$:
\[ \lim_{b \searrow R} v(b) = R. \]

\bigskip
\noindent\textbf{Step 3: The Contradiction.}
We analyze the behavior of buyers with valuations in the neighborhood $v \in [R, R+\epsilon)$.
From Step 2, we established that any buyer choosing a bid $b > R$ must bid truthfully ($b(v)=v$) to satisfy the First Order Condition.
This imposes a strict constraint: the only candidates who could possibly submit bids in the interval $(R, R+\epsilon)$ are buyers with valuations in that same interval, $v \in (R, R+\epsilon)$.

We now check if these specific buyers actually prefer to bid in this region. Consider a buyer with valuation $v = R + \eta$ (where $0 < \eta < \epsilon$). We compare their utility under the two available regimes:

\begin{itemize}
    \item \textbf{Regime A: Bidding in the Price-Taking Region ($b > R$).}
    By the truthfulness condition, the optimal bid in this region is $b = R + \eta$.
    The buyer pays $P=R$ for all infra-marginal trades ($s \le R$).
    The expected utility is approximately:
    \[ U_{\text{right}} \approx (v - R) H_S(R) = \eta \cdot H_S(R). \]
    As $\eta \to 0$, this utility vanishes.

    \item \textbf{Regime B: Bidding in the Monopsony Region ($b < R$).}
    The buyer can deviate to a lower bid $b' = R - \delta$ (for some fixed, small $\delta > 0$).
    Since the support of asks extends below $R$ (from Step 1), trade volume $H_S(b')$ remains strictly positive.
    Crucially, the price on \textit{every} trade drops from $R$ to $b'$.
    The expected utility is:
    \[ U_{\text{left}} = (v - b') H_S(b') = (R + \eta - (R - \delta)) H_S(R - \delta) = (\delta + \eta) H_S(R - \delta). \]
    Since $\delta$ is fixed and independent of $\eta$, this utility is strictly bounded away from zero.
\end{itemize}

\textbf{Conclusion:}
For sufficiently small $\eta$, we have $U_{\text{left}} \gg U_{\text{right}}$.
Thus, buyers with valuations $v \in (R, R+\epsilon)$ strictly prefer to bid below $R$ rather than bid truthfully just above $R$.
This leads to a contradiction:
\begin{enumerate}
    \item The continuity assumption requires that bids cover the neighborhood of $R$.
    \item The truthfulness condition requires that bids in $(R, R+\epsilon)$ can only come from types $v \in (R, R+\epsilon)$.
    \item The incentive analysis shows that types $v \in (R, R+\epsilon)$ refuse to bid there, choosing $b < R$ instead.
\end{enumerate}
Therefore, the interval of bids $(R, R+\epsilon)$ is empty. This gap contradicts the assumption that the bidding strategy $b(v)$ is continuous and strictly increasing.
\end{proof}

\subsection{Proof of \Cref{proposition:universal_stability}}\label{ssec:ProofAuxiliary}

\begin{proof}
\textbf{Part 1: Uniform $L^p$-boundedness.}
Fix an admissible policy $\pi$.
By \cref{assum:transitiondynamics}, the realized feasible interval is $I_t = A_t P_{t-1}+B_t + [U_t,V_t]$. Since $P_t\in I_t$ almost surely, we can bound the magnitude of the action process as:
\[
|P_t| \;\le\; |A_t|\,|P_{t-1}| \;+\; Z_t,
\]
where $Z_t := |B_t|+\max\{|U_t|,|V_t|\}$. By \Cref{assum:stability}, $\mathbb{E}[Z_t^p] < \infty$.

We apply the elementary convexity inequality: for any $x, y \ge 0$, $p \ge 1$, and $\lambda \in (0,1)$,
\[
(x+y)^p \;\le\; \lambda^{1-p} x^p \;+\; (1-\lambda)^{1-p} y^p.
\]
Applying this to $|P_t|$, we obtain:
\[
|P_t|^p \;\le\; \lambda^{1-p} |A_t|^p |P_{t-1}|^p \;+\; (1-\lambda)^{1-p} Z_t^p.
\]
Taking conditional expectations given $\mathcal{F}_{t-1}$ and using that $(A_t, Z_t)$ are independent of the past:
\[
\mathbb{E}^\pi[|P_t|^p \mid \mathcal{F}_{t-1}]
\;\le\;
\lambda^{1-p} \underbrace{\mathbb{E}[|A_t|^p]}_{=: \alpha} |P_{t-1}|^p
\;+\;
(1-\lambda)^{1-p} \mathbb{E}[Z_t^p].
\]
By \Cref{assum:stability}, $\alpha < 1$. Define the function $f(\lambda) = \lambda^{1-p}\alpha$. Since $\lim_{\lambda \to 1^-} f(\lambda) = \alpha < 1$, there exists a $\lambda \in (0,1)$ sufficiently close to 1 such that
\[
\rho \;:=\; \lambda^{1-p}\alpha \;<\; 1,
\qquad
K \;:=\; (1-\lambda)^{1-p}\mathbb{E}[Z_t^p] \;<\; \infty.
\]
Taking total expectations yields the uniform recursion:
\[
\mathbb{E}^\pi[|P_t|^p] \;\le\; \rho\,\mathbb{E}^\pi[|P_{t-1}|^p] \;+\; K.
\]
Iterating this recursion gives $\mathbb{E}^\pi[|P_t|^p] \le \rho^t \mathbb{E}[|P_0|^p] + \frac{K}{1-\rho}$, implying $\sup_{t \ge 0} \mathbb{E}^\pi[|P_t|^p] < \infty$.

\medskip
\textbf{Part 2: Existence of an invariant measure.}
Fix a deterministic stationary Markov policy $\mu$ and let $\nu_t := \mathcal{L}(P_t)$.
By Part 1, $M_p := \sup_{t\ge 0}\int |p|^p\,\nu_t(dp) < \infty$.
By Markov's inequality, for any $R>0$:
\[
\nu_t\big(\{|p|>R\}\big)
\;\le\;
\frac{1}{R^p}\int |p|^p\,\nu_t(dp)
\;\le\;
\frac{M_p}{R^p}.
\]
This vanishes uniformly as $R\to\infty$, so $\{\nu_t\}_{t\ge 0}$ is tight. By Prokhorov's theorem, the sequence of Ces\`aro averages $\bar\nu_T = \frac{1}{T}\sum_{t=0}^{T-1}\nu_t$ is tight and admits a weakly converging subsequence $\bar\nu_{T_n} \Rightarrow \psi_\mu$. The standard Krylov--Bogolyubov argument (utilizing the Feller property of the transition kernel $K_\mu$) confirms that $\psi_\mu$ is an invariant probability measure.

Finally, to bound the $p$-th moment under $\psi_\mu$, integrate the conditional inequality from Part 1 with respect to the invariant measure:
\[
\int \mathbb{E}[|P_t|^p \mid P_{t-1}=x] \,\psi_\mu(dx) \;\le\; \rho \int |x|^p \psi_\mu(dx) \;+\; K.
\]
Using invariance ($\int \mathbb{E}[|P_t|^p \mid P_{t-1}=x] \psi_\mu(dx) = \int |x|^p \psi_\mu(dx)$), we obtain:
\[
\int |x|^p\,\psi_\mu(dx) \;\le\; \frac{K}{1-\rho} \;<\; \infty.
\]
This completes the proof.
\end{proof}

\end{document}